\documentclass{custom-jocg-2023--page-numbers}
\usepackage{amsmath}
\usepackage{bm}
\usepackage{amssymb}
\usepackage{amsthm}

\usepackage{thmtools}
\usepackage{thm-restate}

\usepackage{ifthen}
\newboolean{enable-backrefs} 
\setboolean{enable-backrefs}{true} 
\newcommand{\backrefnotcitedstring}{\relax}
\newcommand{\backrefcitedsinglestring}[1]{(Cited on page~#1.)}
\newcommand{\backrefcitedmultistring}[1]{(Cited on pages~#1.)}
\ifthenelse{\boolean{enable-backrefs}}%
{%
		\PassOptionsToPackage{hyperpageref}{backref}
		\usepackage{backref} 
		   \renewcommand*{\backref}[1]{}  
		   \renewcommand*{\backrefalt}[4]{
		      \ifcase #1 %
		         \backrefnotcitedstring%
		      \or%
		         \backrefcitedsinglestring{#2}%
		      \else%
		         \backrefcitedmultistring{#2}%
		      \fi}%
}{\relax}    

\usepackage{footnotebackref}

\usepackage[capitalise]{cleveref}
\usepackage[T1]{fontenc}
\usepackage[boxruled]{algorithm2e}
\usepackage{subcaption}
\usepackage{graphicx}
\usepackage{soul,color}
\usepackage{multirow}

\usepackage[misc]{ifsym}

\newcommand{\etal}{et~al.\xspace}

\theoremstyle{definition}
\newtheorem{observation}{Observation}
\theoremstyle{definition}
\newtheorem{remark-type-2}{Remark}

\theoremstyle{plain}

\theoremstyle{definition}
\newtheorem{definition}{Definition}
\theoremstyle{plain}

\theoremstyle{plain}

\title{%
  \MakeUppercase{Approximating the Smallest \texorpdfstring{$k$}{k}-Enclosing Geodesic Disc in a Simple Polygon}%
  \thanks{This research was funded in part by the Natural Sciences and Engineering Research Council of Canada (NSERC). A preliminary version of this paper appeared in WADS $2023$ \cite{DBLP:conf/wads/BoseDD23}. A full version of this paper also appeared as a chapter in Anthony D'Angelo's thesis \cite{anthonydangelo2023thesis}. This paper has undergone minor corrections since publication.}
}

\author{%
  Prosenjit~Bose,%
  \thanks{\affil{School of Computer Science, Carleton University, Ottawa, Canada}, 
          \email{jit@scs.carleton.ca}}\,
  Anthony~D'Angelo,%
  \thanks{\affil{Carleton University, Ottawa, Canada},
          \email{anthony.dangelo@carleton.ca}}\,
  Stephane~Durocher%
  \thanks{\affil{University of Manitoba, Winnipeg, Canada}, 
          \email{stephane.durocher@umanitoba.ca}}
}

\hypersetup{
pdftitle={Approximating the Smallest \texorpdfstring{$k$}{k}-Enclosing Geodesic Disc in a Simple Polygon},
pdfauthor={Prosenjit~Bose, Anthony~D'Angelo, Stephane~Durocher},
pdfkeywords={Minimum / smallest enclosing circle / disc\texorpdfstring{$\cdot$}{,} Geodesic\texorpdfstring{$\cdot$}{,} Simple polygon\texorpdfstring{$\cdot$}{,} \texorpdfstring{$2$}{2}-approximation\texorpdfstring{$\cdot$}{,} Computational geometry\texorpdfstring{$\cdot$}{,} SKEG disc\texorpdfstring{$\cdot$}{,} \texorpdfstring{$2$}{2}-SKEG disc},
pdfsubject={Computer Science}
}

\begin{document}


\maketitle

\begin{abstract}

We consider the problem of finding a geodesic disc of smallest 
radius containing at least 
$k$ points from a set of $n$ points in a simple polygon that has 
$m$ vertices, 
$r$ of which are reflex vertices.
We refer to such a disc as a SKEG disc.
We present an algorithm to compute a SKEG disc 
using higher-order geodesic Voronoi 
diagrams with worst-case time 
$O(k^{2} n + k^{2} r + \min(kr, r(n-k)) + m)$
ignoring polylogarithmic factors.

We then present two $2$-approximation algorithms that
find a geodesic disc containing at least $k$ points
whose radius is at most twice that of a SKEG disc.
The first algorithm computes a $2$-approximation 
with high probability
in 
$O((n^{2} / k) \log n \log r + m)$ worst-case
time with $O(n + m)$ space.
The second algorithm runs in $O(n \log^{2} n \log r  + m)$ expected time 
using $O(n + m)$ expected space, independent of $k$.
Note that the first algorithm is faster when $k \in \omega(n / \log n)$.
\end{abstract}



\section{Introduction}
\label{sec: intro}
The \emph{smallest enclosing disc} problem\footnote{Also known as the \emph{minimum enclosing disc} or \emph{minimum enclosing circle} problem.}
takes as input a set $S$ of $n$ points in the plane 
and returns the smallest Euclidean disc
that contains $S$.
This can be solved in $O(n)$ expected time \cite{welzlDisks} and
$O(n)$ worst-case time \cite{DBLP:journals/siamcomp/Megiddo83a}.
The \emph{smallest $k$-enclosing disc} problem is a generalization that asks 
for a smallest disc that
contains at least 
$k \leq |S|$ points\footnote{In this paper, we use the notation $|Z|$ 
to denote the number of points in $Z$ if $Z$ is a point set, or the number of 
vertices of $Z$ if $Z$ is a face or a polygon.} of $S$, 
for any given $k$, and has been well-studied 
\cite{DBLP:journals/jal/AggarwalIKS91,DBLP:journals/jal/DattaLSS95,DBLP:journals/comgeo/EfratSZ94,DBLP:journals/dcg/EppsteinE94,har2005fast,DBLP:journals/ipl/Matousek95,DBLP:journals/dcg/Matousek95-few-violated}.
It is conjectured that an exact algorithm that computes
the smallest $k$-enclosing disc in the plane
requires $\Omega(nk)$ time \cite[\textsection $1.5$]{har2011geometric}.

Agarwal et al.\ \cite{DBLP:journals/jal/AggarwalIKS91} gave an algorithm
to find a set of $k$ points in $S$ with minimum diameter 
(i.e., the selected set minimizes the maximum 
distance between any two points in the set, which 
differs from the
 smallest $k$-enclosing disc)
in $O(k^{5/2}n\log k +n \log n)$ time. 
Efrat, Sharir, and Ziv \cite{DBLP:journals/comgeo/EfratSZ94} used parametric search to 
compute the smallest $k$-enclosing disc
using $O(nk)$ space in $O(nk\log^2 n)$ time, 
and using $O(n\log n)$ space 
in $O(nk\log^2 n \log (n / k))$ time. 
Eppstein and Erickson \cite{DBLP:journals/dcg/EppsteinE94} 
showed how to compute a smallest $k$-enclosing disc
in $O(n\log n + nk\log k)$ time using $O(n\log n + nk + k^2\log k)$ space. 
Datta et al.\ \cite{DBLP:journals/jal/DattaLSS95} were able to 
improve on the result of Eppstein and Erickson
by reducing the space
to $O(n+k^2\log k)$.

Matou\v{s}ek \cite{DBLP:journals/ipl/Matousek95} presented an algorithm that first computes a 
constant-factor approximation\footnote{A $\beta$-approximation means that 
the disc returned has a radius at most $\beta$ times the radius of an optimal solution.} 
in $O(n\log n)$ time and $O(n)$ space
(recently improved to $O(n)$ expected-time for a $2$-approximation
that uses $O(n)$ expected space \cite{har2005fast}),
and then uses that approximation to seed an algorithm
for solving the problem exactly in $O(n\log n + nk)$ expected time using $O(nk)$ space
(which Datta et al.\ \cite{DBLP:journals/jal/DattaLSS95} improved to $O(n+k^2)$ space)
or $O(n\log n + nk\log k)$ expected time using
$O(n)$ space
(recently improved to $O(nk)$ expected time using 
$O(n+k^2)$ expected space \cite{DBLP:journals/jal/DattaLSS95,har2005fast}).
Matou\v{s}ek \cite{DBLP:journals/dcg/Matousek95-few-violated} also presented an algorithm 
for computing the smallest disc that contains all but at most $q$ of $n$
points 
in $O(n\log n + q^3n^{\epsilon})$ time for $\epsilon$  
``a positive constant that can be made arbitrarily small by adjusting the parameters
of the algorithms; multiplicative constants in the $O()$ notation may depend on $\epsilon$''
\cite{DBLP:journals/dcg/Matousek95-few-violated}.

In this paper we generalize the
smallest $k$-enclosing disc
problem to simple polygons using the \emph{geodesic metric}.
This means the Euclidean distance function has been replaced by the 
\emph{geodesic} distance function.
The \emph{geodesic distance} between two points $a$ and $b$
in a simple polygon $P$ is the length
of the shortest path between $a$ and $b$ that lies completely inside
$P$. 
When $P$ is a simple polygon with a finite set of vertices,
a shortest path from $a$ to $b$ in $P$ is a polyline whose length is
the sum of the Euclidean lengths of its edges.
Let this shortest path starting at $a$ and ending at $b$ be 
denoted $\Pi(a,b)$ and let its length be denoted
as $d_g(a,b)$.
For any two points $a$ and $b$ of $P$, the vertices
on $\Pi(a, b)$ are either $a$, $b$, or reflex vertices of $P$
\cite{Aronov89,DBLP:journals/networks/CheinS83,DBLP:journals/networks/LeeP84}.
A \emph{geodesic disc} $D(c, \rho)$ 
of radius $\rho$ centred at $c\in P$ is the set of all points 
in $P$ whose geodesic distance to $c$ is at most $\rho$.
The main focus of our article is the following problem.

\begin{description}
	\item[Smallest $k$-Enclosing Geodesic (\emph{SKEG}) Disc Problem] ~\\
	Consider a simple polygon $P_{in}$ 
	defined by a sequence of $m$ vertices in $\mathbb{R}^2$, 
	$r$ of which are reflex vertices,
	and a set $S$ 
	of $n$ points of $\mathbb{R}^2$ contained in 
	$P_{in}$.\footnote{When we refer to 
	a point $p$ being in a polygon $P$, we mean that $p$ is in the interior of $P$ 
	or on the boundary, $\partial P$.}
	Find a \emph{SKEG disc}, i.e., a geodesic disc of minimum radius $\rho^*$ 
	in $P_{in}$ that contains at least $k$ points of $S$.
\end{description}

	We make the general position assumptions that no two points of $S$ are 
	equidistant to a vertex of $P_{in}$, and no four points of $S$ are
	geodesically co-circular.
	Under these assumptions, a SKEG disc contains
	exactly $k$ points.
	Let $D(c^*, \rho^*)$ be a SKEG disc for the 
    points of $S$ in $P_{in}$.
    For convenience, at times we will refer to this as simply $D^*$.	
	A $k$-enclosing geodesic disc (\emph{KEG} disc) is a geodesic disc 
	in $P_{in}$ that contains
    exactly $k$ points of $S$.
    A $2$-SKEG disc is a KEG disc
    with radius at most $2\rho^*$ (i.e., it is a $2$-approximation).
The main result of our article is the following theorem.

\begin{restatable*}{theorem}{NewMainResultTheorem}
\label{theorem: new main result}
If $k \in O(n / \log n)$,
\textbf{\ref{alg: main}} 
computes a $2$-SKEG disc 
in expected time $O(n \log^2 n \log r  + m)$ and 
expected space $O(n + m)$, independent of $k$;
if $k \in \omega(n / \log n)$,
\textbf{\ref{alg: main}} 
computes a $2$-SKEG disc
with high probability\footnote{We say an event 
happens with high probability if the probability is at least
$1 - n^{-\lambda}$ for some constant $\lambda$.} 
in $O((n^{2}/k) \log n \log r + m)$ deterministic
time with $O(n + m)$ space.
\end{restatable*}

\subsection{Related Work}
\label{subsec: related}
We are not aware of other work tackling the subject of this paper,
but 
below we highlight
related work on
geodesic discs in polygons.

A region $Q$ is \emph{geodesically convex} relative to a polygon $P$ if for all 
points $u,v \in Q$, the geodesic shortest path from $u$ to $v$ in 
$P$ is in $Q$.
The \emph{geodesic convex hull} $CH_g$ of a set of points $S$ in a polygon 
$P$ is the intersection of all geodesically convex regions in $P$ that 
contain $S$. 
The geodesic convex hull of $n$ points in a simple $m$-gon
can be computed in $O( n\log (mn) + m)$ time 
(or $O(n\log (mn))$ time once the shortest-path data structure
of Guibas and Hershberger \cite{GUIBAS1989126,HERSHBERGER1991231} 
is built) using $O(n+m)$ space
\cite{GUIBAS1989126,toussaint1989computing}.
The expression of the runtime can be simplified 
to $O(n\log n + m)$ 
using the following argument \cite{DBLP:journals/algorithmica/Liu20,DBLP:conf/soda/Oh19,DBLP:conf/compgeom/000121}.
If $n\leq m / \log m$, then we have
$    n\log m \leq (m / \log m) \log m = m$. 
If $m / \log m < n$, then we have
$    \log(m / \log m) < \log n$
which implies that $\log m$ is $O(\log n)$.

The \emph{geodesic centre} problem asks for a smallest 
enclosing geodesic disc
that lies in the polygon and encloses all vertices of the polygon 
(stated another way, a point that minimizes the geodesic distance to the farthest point).
This problem is well-studied \cite{DBLP:journals/dcg/AhnBBCKO16,asano1987computing,DBLP:conf/cgi/BoseT96,DBLP:journals/dcg/PollackSR89,toussaint1989computing}
and can be solved in $O(m)$ time and space
\cite{DBLP:journals/dcg/AhnBBCKO16}.
The geodesic centre for polygonal domains\footnote{As defined in Bae et al.~\cite{DBLP:journals/comgeo/BaeKO19}, a \emph{polygonal domain} with $h$ holes and $m$ corners is a connected and closed subset of $\mathbb{R}^{2}$ having $h$ holes whose boundary consists of $h+1$ simple closed polygonal chains of $m$ total line segments.} 
has also been studied 
\cite{DBLP:journals/comgeo/BaeKO19,DBLP:journals/jocg/Wang18} and can 
be solved in $O(m^{11}\log m)$ time \cite{DBLP:journals/jocg/Wang18}.
The geodesic centre problem for simple polygons has been generalized 
to the geodesic $j$-centre problem\footnote{The $j$-centre problem
in the Euclidean plane, and thus also
in our geodesic setting, is NP-hard \cite{DBLP:journals/siamcomp/MegiddoS84}.} 
where we seek a set of $j$ geodesic discs of minimum radii such that all 
vertices of the simple polygon are contained in some disc, though
only $j=2$ has been studied so far 
\cite{DBLP:journals/comgeo/OhCA18,DBLP:journals/corr/Vigan13}.
The current best algorithm by Oh et al.\ 
\cite{DBLP:journals/comgeo/OhCA18} runs
in $O(m^2\log^2 m)$ time and $O(m)$ space.
The geodesic $j$-centre problem has been generalized to finding the geodesic
$j$-centre of a set of points $S$ inside a simple polygon, though we still only have
results for $j=1$ and $j=2$.
For $j=1$, the problem can be solved by finding the geodesic centre of 
the weakly simple polygon formed by the geodesic convex hull 
of the point set $S$ \cite{DBLP:journals/dcg/AronovFW93},
 and thus runs in 
 $O(n\log n + m)$ time. 
For this approach, 
computing the solution 
is dominated by the time to compute $CH_g$. 
Oh et al.\ \cite{DBLP:journals/comgeo/OhBA19} 
recently gave an algorithm for $j=2$ that runs in
$O(n(m+n)\log^3 (m+n))$ time.
There is also the notion of the \emph{geodesic edge centre} of a simple polygon 
(which is a point in the polygon that minimizes the maximum geodesic distance to any edge of the polygon)
for which Lubiw and Naredla have presented a linear-time algorithm to compute \cite{DBLP:conf/compgeom/LubiwN23}.

Vigan \cite{DBLP:journals/corr/Vigan13} 
has worked on packing and covering a simple polygon with geodesic 
discs.
He provides a $2$-approximation for the maximum cardinality 
packing of geodesic unit discs in simple polygons in
$O(Y(m+Y)\log^2 (m +Y))$ time,
where $Y$ is the output size, and shows that finding a packing of geodesic 
unit discs is NP-hard in polygons with holes;
he shows that a minimum-size covering of
a polygon with holes using geodesic unit discs
is NP-hard;
he shows that covering a polygon using $j$ geodesic discs whose maximal
radius is minimized is NP-hard and provides a $2$-approximation algorithm
that runs in $O(j^2(m+j)\log (m+j))$ time;
he shows that packing $j$ geodesic discs whose minimum radius is maximized 
into a polygon is NP-hard and gives an 
$O(j^2(m+j)\log (m+j))$-time $4$-approximation algorithm;
lastly, he gives an $O(m^8\log m)$-time exact algorithm for covering a simple
polygon with two geodesic discs of minimum maximal radius.

Rabanca and Vigan \cite{DBLP:journals/corr/RabancaV14}
worked on covering the boundary of a simple 
polygon with the minimum number of geodesic unit discs 
and gave an $O(m\log^2 m + j)$-time $2$-approximation algorithm, 
where $j$ is the number of discs, and they also showed that if the perimeter 
$L$ of the polygon is such that $L \geq m^{1+\delta}, \delta > 0$, then a 
simple $O(m)$-time algorithm achieves an approximation ratio that goes to 
$1$ as $L/m$ goes to infinity.

Borgelt, van~Kreveld, and Luo \cite{DBLP:journals/ijcga/BorgeltKL11}
studied geodesic discs in simple polygons
to solve the following clustering problem:
given a set $S$ of $n$ points inside a simple polygon $P$ with $m$
vertices, determine all subsets of $S$ of size at least $k$ for which
a point $q$ in $P$ exists that has geodesic distance at most $\rho$ to all
points in the subset.
The input parameters are $k$ and $\rho$.
They solve their problem by generating
the boundaries of radius-$\rho$ 
geodesic discs centred at the points of $S$ and finding points that lie
in at least $k$ discs.
They provide an output-sensitive algorithm
to compute the set of geodesic discs centred at the point sites with the 
given radius that runs in $O((m+ (Ym)^{(2 / 3)}+Y)\log^c m)$ time
for some constant $c$ and output size $Y$.
With $m$ vertices of the polygon and $n$ input points, 
each of the $n$ discs could have $O(m)$ arcs/pieces
due to reflex vertices as well as the edges of the polygon
that contain points closer than $\rho$ to the centre.
Thus, the upper bound on $Y$ is $O(nm)$.
After computing the boundaries of the discs individually,
they then build the arrangement of the discs, count the depth of
a cell of the arrangement, then traverse the arrangement 
updating the depth as they enter new cells.
The runtime to solve the clustering problem 
also comes to depend on the number of intersections
in the arrangement, which is
$O(Y + n^2)$ \cite{DBLP:journals/ijcga/BorgeltKL11}.
The runtime becomes 
$O((m+ (Ym)^{(2 / 3)}+Y)\log^c m + Y\log Y + n^2)$.

Dynamic $k$-nearest neighbour queries were studied by
de Berg and Staals \cite{DEBERG2023101976}.
They presented a static data structure for geodesic 
$k$-nearest neighbour queries for $n$ sites in a simple $m$-gon
that is built in $O(n(\log n \log^2 m + \log^3 m))$ expected time 
using $O(n\log n \log m + m)$ expected space
and answers queries in 
$O(\log (n+m)\log m + k\log m)$ expected time.\footnote{
We note that depending on the relations of the values $m$, $r$, and $n$ to each other,
there may be situations in which our algorithms may be improved by polylogarithmic factors
by using the 
$k$-nearest neighbour query data structure.
Details of these trade-offs are discussed in \cref{subsec: knn disc}.}

If $P_{in}$ has no reflex vertices, it is a convex polygon
and the SKEG
disc problem is solved by the algorithm for planar instances which uses 
a grid-refinement strategy.
This works in the plane because 
$\mathbb{R}^2$ with
the Euclidean metric, $d_{e}(\cdot,\cdot)$, is a doubling metric space,
meaning that for any point $c$ in the space $\mathbb{R}^2$ and any radius $\rho>0$, 
the closed disc
$B(c, \rho) = \{ u \in \mathbb{R}^2 : d_{e}(u, c) \leq \rho\}$ can be covered
with the union of $O(1)$ closed discs of radius $\rho / 2$
\cite{metricBook} (i.e., unlike the geodesic metric, 
the Euclidean metric has \emph{bounded doubling dimension}).
The advantage of bounded doubling dimension to the gridding strategy is that 
a cell of the grid has a constant number of neighbours that need to be
considered when processing the cell, allowing us 
to locate $O(k)$ candidate neighbour points when processing a point \cite{har2011geometric,har2005fast}.
It is straightforward to prove 
(using an example such as the one in \cref{fig: bounded dd} from \cref{app: doubling geodesic})
that the geodesic metric does not have bounded doubling dimension
(its doubling dimension is $\Theta(r)$ in a simple polygon with $r$ reflex vertices).
Thus, it is not clear whether 
a gridding strategy can be used with the geodesic
metric to get efficient algorithms.
The straightforward
adaptation 
to simple polygons
of the expected linear-time
approximation algorithm that relies on grids \cite{har2005fast}
no longer seems to run in expected linear time;
in a simple polygon with $r$ reflex vertices,
a given cell of a traditional grid approach could have $\Omega(r)$
neighbours to be considered when processing a cell,
and it is not straightforward to process a cell that contains
a reflex vertex.

Another difficulty of the geodesic metric is that 
for two points $u$ and $v$ of $S$ on opposite sides of 
a given chord, their geodesic bisector (formed by concatenating their 
bisector and hyperbolic arcs) can cross the chord 
more than once.
See \cref{fig: bisector r crossings,fig: bisector r crossings zoomed a,fig: bisector r crossings zoomed b}
in \cref{app: bisector crosses chord}.

\subsection{An Exact Algorithm}
\label{subsec: exact}
In this section we present an exact algorithm for the SKEG disc
problem that uses higher-order geodesic Voronoi diagrams to find the
exact solution. 
Omitted details appear in \cref{app: vd}.

Rather than working with our $m$-gon,
we use the polygon simplification algorithm of 
Aichholzer \etal\ \cite{aichholzer2014geodesic} to transform $P_{in}$ into
a simple polygon $P \supseteq P_{in}$ of size $O(r)$ in $O(m)$ time and space
such that $P$ preserves the visibility of points in $P_{in}$, 
as well as their shortest paths.
The reflex vertices
in $P_{in}$ also appear in $P$.
We again assume that $S$ is in general position with respect to $P$.
The polygon simplification allows the running time of the 
algorithm to depend on the number of \emph{reflex} vertices of $P_{in}$ rather 
than the total number of vertices of $P_{in}$.
Thus, the algorithm runs faster on polygons with fewer 
reflex vertices.

For $k \leq n-1$ we use the shortest-path data structure of
Guibas and Hershberger \cite{GUIBAS1989126,HERSHBERGER1991231}.
This shortest-path data structure can be built in
$O(r)$ time and space
(since we have linear-time polygon triangulation 
algorithms \cite{DBLP:conf/compgeom/AmatoGR00,Chazelle1991})
and, given any two query points in the polygon,
returns a tree of 
$O(\log r)$ 
height in $O(\log r)$ time
representing the shortest path
between the two query points as well as the length of this path.
The in-order traversal of this tree gives the shortest path.
The data structure uses additional $O(\log r)$ space to build the 
result of the query (i.e., the tree) by linking together precomputed
structures.
Given the result of a query, we can perform a search through
this tree to find the midpoint of the shortest path in 
$O(\log r)$  time
\cite[Lemma $3$]{DBLP:conf/compgeom/000121},
or traverse the resulting tree to a leaf to get 
the first or last edge of the path in $O(\log r)$ time.

Higher-order Voronoi diagrams have been considered to
solve the smallest $k$-enclosing disc problem in the plane
\cite{DBLP:journals/jal/AggarwalIKS91,DBLP:journals/comgeo/EfratSZ94}.
This approach can be generalized to a point set $S$ contained in
simple polygons,
but it requires computing the order-$k$ geodesic Voronoi diagram
(OKGVD),\footnote{Stated briefly, the order-$k$ Voronoi diagram
is a generalization of the Voronoi
diagram such that each face is the locus of points whose
$k$ nearest neighbours are the $k$ points of $S$ 
associated with (i.e., that define) the face.}
or the order-($k-1$) and order-$(k-2)$ diagrams.

The order-$1$ and the order-$(n-1)$ geodesic Voronoi diagrams
(a.k.a., the \emph{nearest-point} and \emph{farthest-point} geodesic Voronoi diagrams, respectively)
have complexity $\Theta(n + r)$ \cite{Aronov89,DBLP:journals/dcg/AronovFW93}.
There is an algorithm for the geodesic nearest-point Voronoi diagram
that runs in $O( n\log n + r)$ time and uses 
$O(n\log n + r)$
space 
\cite{DBLP:conf/soda/Oh19}, and one for the
farthest-point diagram that runs in  
$O( n\log n + r)$ time and uses $\Theta(n + r)$ space 
\cite{DBLP:conf/compgeom/Barba19,DBLP:conf/compgeom/000121}.

Oh and Ahn \cite{Oh2019} showed that the complexity 
of the 
OKGVD
of $n$ points in a simple $r$-gon is
$\Theta(k(n-k) + \min(kr, r(n-k)))$.
The fastest algorithms to compute the OKGVD
are the one by Oh and Ahn~\cite{Oh2019} that runs in time 
$O(k^2n\log n \log^2 r + \min(kr, r(n-k)))$,
and the one by Liu and Lee~\cite{DBLP:conf/soda/LiuL13} that runs in
$O(k^2n\log n + k^2 r \log r)$  
time.
They both use $\Omega(k(n-k) + \min(kr, r(n-k)))$ 
space (which is 
$\Omega(n+r)$ 
for constant $k$ or $k$ close to $n$ 
and 
$\Omega(kn + kr) = \Omega(n^2+nr)$ 
for $k$ a constant fraction of $n$). 

When using the order-$k$ diagram to find a SKEG disc, the approach is to traverse the diagram
and, for each face, compute the geodesic centre of the $k$ points that define
the face (either by using the approach of Oh and Ahn~\cite[Lemma $3.9$]{Oh2019},
or by computing $CH_g$ for the points and then computing the geodesic centre
of $CH_g$).
While computing the geodesic centre for the $k$ points that define a face,
we also find one of these $k$ points farthest from this centre.
As such, we can then get the radius of a SKEG disc for these points in $O(\log r)$
time using shortest-path queries.
A solution to the SKEG disc problem is then the smallest of these discs.

\begin{restatable}{remark-type-2}{genPosRemark}
\label{rem: gen pos}
Assuming general positioning,
a 
SKEG
disc has either two or three points of $S$
on its boundary.
\end{restatable}

The second approach for computing a SKEG disc with geodesic Voronoi diagrams
computes the order-($k-1$) and order-$(k-2)$ diagrams.
If there are three points of $S$ on the boundary of a SKEG disc, then
the centre of a SKEG disc is a Voronoi vertex of the order-$(k-2)$ diagram
and the radius is the geodesic distance from that Voronoi vertex to the 
points of $S$ that define it.
Given the order-$(k-2)$ diagram, we traverse it and compute a SKEG disc
for the $k$ points defining each Voronoi vertex in $O(\log r)$ time each 
(i.e., the three faces adjacent to a given Voronoi vertex will be defined by $k$ 
distinct points of $S$, and the SKEG disc for those points will be the geodesic disc
centred at the Voronoi vertex and whose radius is the distance to any one of the three points
of $S$ that defined the vertex).
If there are two points of $S$ on the boundary of a SKEG disc, 
then the centre of a SKEG disc lies along a Voronoi edge of the order-($k-1$) diagram,
more specifically at the midpoint of the shortest path between the two points of $S$ that defined
the geodesic bisector defining the Voronoi edge.
Given the order-$(k-1)$ diagram, we traverse it and compute a SKEG disc for 
the $k$ points defining each Voronoi edge in $O(\log r)$ time each 
(i.e., the two faces adjacent to a given Voronoi edge will be defined by $k$
distinct points of $S$, and the SKEG disc for those points will be the geodesic
disc centred at the midpoint of the aforementioned shortest path 
and whose radius is half the length of said shortest path).
The smallest computed disc is a SKEG disc for the points of $S$.

When $k = n-1$, it is more efficient to use the order-$k$
diagram to solve the SKEG disc problem.
Otherwise, for $1 < k < n-1$,
it is more efficient to use the order-($k-1$) and 
order-$(k-2)$ diagrams to solve the SKEG disc problem.

\begin{restatable}{theorem}{okvdExactApproach}
\label{theorem: exact answer order k VDs}
Given an integer $1 < k \leq n$,
a SKEG disc of the points of $S$
can be computed in: \\
\begin{tabular}{|c|l|l|}
\hline
\bm{$k$} & \textbf{Time} & \textbf{Space}            \\ 
\hline 
\hline
$n$                                         & $O(n\log n + m)$          &    $\Theta(n+m)$            \\ 
\hline
\multirow{5}{*}{$n-1$}   & if $k\log k \log^2 r \in O(r)$       &   \\
                        & \quad \textbullet\ $O(nr\log n \log^2 r +  m)$     &    $\Theta(n+m)$  \\
                        & else  &  \\
                        & \quad \textbullet\ $O(n^2\log n + n^2\log r + nr\log n + nr\log r + r^2 +  m)$  &   \\ 
\hline
$2$                                         & $O(n\log n + r\log r + m)$          &    $O(n\log n +m)$            \\ 
\hline
\multirow{5}{*}{else}   &  if $n\log n \log r \in O(r)$                     &   $\Omega(k(n-k)$     \\ 
                        &  \quad \textbullet\ $O(k^2 n\log n \log^2 r +  \min(rk,r(n-k)) \log r + m)$    & \quad $+ \min(rk, r(n-k))$ \\
                        &   else                                        & \quad $+ m)$  \\
                        &  \quad \textbullet\ $O(k^2 n\log n + k^2 r\log r + k(n-k)\log r +$  & \\
                        & \qquad \quad $\min(kr, r(n-k)) \log r + m)$   & \\
\hline
\end{tabular}

Ignoring polylogarithmic factors,
the worst-case runtime 
for $k=\{2, n\}$ is $O(n + m)$;
for $k=n-1$: if $k\log k \log^2 r \in O(r)$ is
$O(nr + m)$, 
else, is
$O(n^2 + nr + r^2 + m)$;
for $k < n-1$: if $n\log n \log r \in O(r)$
is $O(k^{2} n + \min(rk,r(n-k)) + m)$,
and $O(k^{2} n + k^{2} r + \min(kr, r(n-k)) + m)$ otherwise.
\end{restatable}

\begin{proof}
Omitted details appear in \cref{app: vd}.

When $k=n$, the SKEG disc problem is the same as the problem of
computing the geodesic centre of a set of points inside a simple polygon.

We can use the order-$1$ geodesic Voronoi diagram
when $k = 2$.
Considering the preprocessing time, the construction time,
and the time to traverse the diagram performing the 
appropriate $O(\log r)$-time queries, we can solve the SKEG disc problem
in $O(n\log n + r\log r + m)$ time (simplified from $O(n\log n + n\log r + r\log r + m)$).

The farthest-point geodesic Voronoi
diagram is the same as the order-$(n-1)$ diagram.
When $k=n-1$,
it is more efficient to compute a SKEG disc using the farthest-point 
geodesic Voronoi diagram and, 
for each face, 
compute the geodesic centre of the $n-1$ closest points.
Oh and Ahn~\cite[Lemma $3.9$]{Oh2019} presented a method to find
the geodesic centre of $k$ points in 
$O(k\log k \log^2 r)$ time and $O(k+r)$ space.
The other method for computing the appropriate geodesic centres is to 
first compute $CH_g$ for the $k$ points defining each face, and then compute
the geodesic centre of the weakly simple polygon formed by the $CH_g$.
The time for the Oh and Ahn approach (i.e., $O(k\log k \log^2 r)$) always dominates
the time to compute $CH_g$ (i.e., $O(k\log r + k \log k)$).
Thus, the Oh and Ahn approach is only more efficient for 
computing the geodesic centres of the requisite points when its runtime
is dominated by the time to compute the geodesic centre of $CH_g$.
As has already been mentioned, 
the time to compute the geodesic centre of a weakly simple polygon is linear in its
size.
The size of $CH_g$ of the $k$ points of $S$ defining
a face of the OKGVD is $O(r + k)$.
Consequently, the Oh and Ahn~\cite[Lemma $3.9$]{Oh2019} approach is more efficient when 
$k\log k \log^2 r \in O(r)$.
Performing this computation for each face of the OKGVD,
we spend  $O((r+n)\cdot k\log k \log^2 r)$ time over all faces
and $O(r+n)$ space.
Substituting $k=n-1$, noting that $k\log k \log^2 r \in O(r)$
means $n \in O(r)$, 
and considering the preprocessing,
the time we use is	$O(nr\log n \log^2 r + m)$.

Otherwise, for $k=n-1$ it is more efficient to
compute $CH_g$ of the $k$ points of $S$ associated
with each face of the diagram to obtain weakly simple polygons, 
and then compute the geodesic centre of each of
those geodesic convex hulls.
We spend 
$O(r+n)\cdot O(k\log r + k \log k)$
time to compute $CH_g$ of each subset of $k$ points
that defines a face;
then we compute the geodesic centres in
$O(r+n)\cdot O(r+k)$ time
and $O(r+n)$ space.
Substituting
$k = n-1$ and considering the preprocessing cost,
the time used becomes $O(n^2\log n	+  n^2\log r + nr\log n + nr\log r + r^2 +  m)$.

For $1 < k < n-1$, it suffices
to compute the order-$(k-2)$ and order-$(k-1)$ geodesic Voronoi diagrams
and consider: the vertices of the order-$(k-2)$ diagram;
and the midpoints of the shortest paths 
between the points defining edges shared by adjacent cells
of the order-$(k-1)$ diagram.

Oh and Ahn \cite{Oh2019} showed that the complexity 
of the 
OKGVD
of $n$ points in a simple $r$-gon is
$\Theta(k(n-k) + \min(kr, r(n-k)))$.
The fastest algorithms to compute the OKGVD
are the one by Oh and Ahn~\cite{Oh2019} that runs in time 
$O(k^2n\log n \log^2 r + \min(kr, r(n-k)))$,
and the one by Liu and Lee~\cite{DBLP:conf/soda/LiuL13} that runs in
$O(k^2n\log n + k^2 r \log r)$  
time.
They both use $\Omega(k(n-k) + \min(kr, r(n-k)))$ 
space (which is 
$\Omega(n+r)$ 
for constant $k$ or $k$ close to $n$ 
and 
$\Omega(kn + kr) = \Omega(n^2+nr)$ 
for $k$ a constant fraction of $n$). 
The algorithm of Oh and Ahn~\cite{Oh2019}
is more efficient when $n\log n \log r \in O(r)$,
which means the algorithm of Liu and Lee~\cite{DBLP:conf/soda/LiuL13} 
is more efficient the rest of the time.
Using shortest-path queries
we traverse the diagrams in 
$O(k(n-k) + \min(kr, r(n-k)))\cdot O(\log r)$ time
while
comparing the candidates from the edges
and vertices.
Thus, for $n\log n \log r \in O(r)$, 
the time used including preprocessing is $O(k^2n\log n \log^2 r + \min(kr, r(n-k))\log r + m)$.
Otherwise, the time including preprocessing is
$O(k^2 n \log n + k^2 r \log r + k(n-k)\log r +  \min(kr, r(n-k)) \log r + m)$.
\end{proof}

The main result of our paper is an improvement in the 
runtime for computing a SKEG disc (\cref{theorem: exact answer order k VDs}),
but it comes
at the expense of a $2$-approximation. 
This is summarized in 
\cref{theorem: new main result}.
The runtime of \cref{theorem: new main result} 
is derived by balancing the runtimes of two algorithms:
\ref{alg: RSAlgo} 
(a random sampling algorithm
described in \cref{sec: random sampling})
and 
\ref{alg: DI}
(a Divide-and-Conquer algorithm described 
in \cref{sec: Divide et Impera}).

Ignoring polylogarithmic factors,
the expected runtime of the approximation 
algorithm of 
\cref{theorem: new main result} is $O(n + m)$,
matching the exact approach for $k = \{2, n\}$.
The approximation algorithm of \cref{theorem: new main result}
is roughly expected to be
faster by a factor of $k$ or $r$ for $k=n-1$, 
and a factor of $k^2$ otherwise. 
For example, if $k \in \Omega(n)$ and $k < n-1$, 
compare $O(n + m)$ to:
$O(n^3 + r + m)$ for $k$ close to $n$, 
or $O(n^3 + nr + m)$ for $k$ a fraction of $n$;
and $O(n^3 + n^2r + m)$.

\cref{section: quickselect merge step} shows
how to use a randomized iterative search as
the merge step in \cref{sec: Divide et Impera}  
to compute a $2$-approximation.
We summarize our results in \cref{sec: concluding remarks}.
\cref{app: doubling geodesic} illustrates an example where the 
doubling dimension in the geodesic setting is proportional to the
number of reflex vertices of the simple polygon being considered;
\cref{app: bisector crosses chord} provides an example where the geodesic
bisector of two points in a simple polygon
crosses a chord of that polygon multiple times;
\cref{app: vd} provides omitted details detailing the exact algorithm
that uses geodesic Voronoi diagrams (i.e., the approach described in this section)
and the calculation of the runtimes presented in \cref{theorem: exact answer order k VDs};
\cref{app: di-algo} provides details omitted from the section describing the Divide-and-Conquer
algorithm presented in this paper (i.e., the algorithm from \cref{sec: Divide et Impera});
\cref{app: knn queries} provides details omitted from the discussion in \cref{subsec: knn disc} about 
using $k$-nearest neighbour queries in our algorithms;
and \cref{app: code} provides pseudocode for the algorithms described 
in the remainder of this paper.

\section{Random Sampling Algorithm}
\label{sec: random sampling}
In this section we present a random sampling algorithm to compute a 
$2$-SKEG disc.
Our first algorithm, \ref{alg: RSAlgo},
uses the following preprocessing (repeated for convenience).

\begin{description}
	\item[Polygon Simplification]  
 	Convert $P_{in}$ into a simplified 
 	polygon $P$ consisting of $O(r)$ vertices using the algorithm of 
 	Aichholzer \etal \cite{aichholzer2014geodesic}. 	
 	The algorithm runs in $O(m)$ time and space and computes a polygon $P$ 
 	such that $P \supseteq P_{in}$, $|P|$ is $O(r)$, and the reflex vertices
 	in $P_{in}$ also appear in $P$.
 	Furthermore, 
 	the shortest path
 	between two points in $P_{in}$ remains unchanged in $P$.
 	As with $P_{in}$, we assume the points of $S$ are in general position 
 	with the vertices of $P$, and no four points of $S$ are geodesically
 	co-circular in $P$.

	\item[Shortest-Path Data Structure] 
	We use the $O(r)$-time and $O(r)$-space algorithm of 
	Guibas and Hershberger \cite{GUIBAS1989126,HERSHBERGER1991231}
	on $P$ to build a data structure
	that gives the length of the shortest path between
	any two points in $P$ in $O(\log r)$ time and space.
	Their algorithm runs
	in linear time since we have linear-time polygon triangulation 
	algorithms \cite{DBLP:conf/compgeom/AmatoGR00,Chazelle1991}.
	Querying the data structure with two points in $P$ 
	(i.e., a \emph{source} and \emph{destination})
	returns a tree of $O(\log r)$ height
	whose in-order traversal is the shortest path in $P$ between the two query points.
	The length of this path is returned with the tree, as is the length of the
	path from the source to each node along the path (which is stored at the
	respective node in the tree).
	The data structure uses additional $O(\log r)$ space to build the 
    result of the query (i.e., the tree) by linking together precomputed
    structures.
	In addition to the length of the shortest path between two points, this
	data structure can provide the first or last edge along the path between
	the two points in $O(\log r)$ additional time by traversing the tree to a leaf.
\end{description}

This completes the preprocessing.
The total time and space spent preprocessing is $O(m)$.
The polygon simplification allows the running time of the 
algorithm to depend on the number of reflex vertices of $P_{in}$ rather 
than its size.
Thus, the algorithm runs faster on polygons with fewer 
reflex vertices.
The shortest-path data structure allows us to quickly determine the geodesic 
distance between two points inside $P$.
\textbf{\ref{alg: RSAlgo}} proceeds as follows:
we repeatedly find the smallest 
KEG
disc
centred on the currently considered point of a computed random sample.
\begin{enumerate}
	\item Compute a random sample of $(n / k)\ln (n)$ points of $S$.
	\item For each point $c$ in the random sample, 
	find its $(k-1)$\textsuperscript{st}-closest 
	point in $S$ using the geodesic distance.
	\item Return the 
    KEG disc of minimum computed radius.
\end{enumerate}

\begin{restatable}{lemma}{RSAlgoProbChoosesGoodPt}
\label{lemma: rsalgo probably chooses a point in opt disc}
\ref{alg: RSAlgo} chooses
one of the $k$ points of $S$ contained in 
a SKEG disc $D^*$ 
with probability at least $1- n^{-1}$. 
\end{restatable}

\begin{proof}
Picking a point $c\in S$ uniformly at random, the probability that
$c$ is one of the $k$ points of $S$ in $D^*$ is
$P(c \in D^*) = k / n$, since each of the $k$ points
in $D^*$ has probability $1 / n$ of being chosen.
Thus we have $P(c \notin D^*) = 1 - (k / n)$.
The probability that 
none of the $k$ points of $S$ contained in $D^*$
is chosen among
the $(n / k)\ln (n)$ randomly chosen points is:

\[
(1 - (k / n))^{(n /k)\ln (n)} \leq  
(\rm{e}^{- (k /n)})^{(n /k)\ln (n)}
=  1 / n
\]
Therefore, 
after selecting $(n / k)\ln (n)$ points of $S$ 
uniformly at random,
we will have chosen one of the $k$ points of $S$
in $D^*$ with probability at least $1 - n^{-1}$.
\end{proof}

\begin{restatable}{lemma}{BFDiscTwoApprox}
\label{lemma: bfdisc 2-approx}
\ref{alg: RSAlgo} 
computes a 
$2$-SKEG disc
if some point $c \in S$ in the random sample
is one of the $k$ points of $S$ contained in
a SKEG disc $D^*$.
\end{restatable}

\begin{proof}
Let $x$ be a point in $D^*$.
By the definition of a disc, we have that $D^* \subset D(x, 2\rho^*)$.
Therefore, the smallest geodesic disc centred at $x$ containing $k$ points 
has a radius that is at most twice $\rho^*$ since $D^*$ contains $k$ points.
\end{proof}

Having established \cref{lemma: bfdisc 2-approx,lemma: rsalgo probably chooses a point in opt disc}, we are ready to prove
\cref{theorem: rs algo}.

\begin{restatable}{theorem}{RSAlgoThm}
\label{theorem: rs algo}
\textbf{RS-Algo} 
computes
a $2$-SKEG disc with 
probability
at least $1 - n^{-1}$
in deterministic time $O((n^{2} / k)\log n \log r + m)$
using $O(n + m)$ space.
\end{restatable}

\begin{proof}
We simplify our polygon $P_{in}$ in $O(m)$ time and space
\cite{aichholzer2014geodesic}.
We create our data structure for $O(\log r)$-time 
shortest-path queries in $O(r)$ time and space 
\cite{GUIBAS1989126,HERSHBERGER1991231}.

The runtime and space of the algorithm
is dominated
by the time and space used by the for-loop
to compute the $(k-1)$\textsuperscript{st}-closest 
neighbour of each of the $(n / k) \ln (n)$ points in the random sample.
Since our preprocessing allows us to compute the distance 
between two arbitrary points in $O(\log r)$ time and space, 
the runtime for 
computing the $(k-1)$\textsuperscript{st}-closest point of $S$ for
one point $c \in S$ in the random sample is $O(n \log r)$
and it uses $O(n + \log r)$ space.
Indeed,
in $O(n)$ time and space we create an array to store distances;
next, we iterate over each of the $O(n)$ points 
in $S \setminus \{c\}$, calculate their distance to $c$ in 
$O(\log r)$ time, 
and store the result in the next free space
in our array;
lastly, we need to extract the $(k-1)$\textsuperscript{st}-smallest value
from our array (this can be done using an $O(n)$-time rank-finding
algorithm~\cite{DBLP:conf/wea/Alexandrescu17,DBLP:journals/jcss/BlumFPRT73,DBLP:books/daglib/0023376}).
Therefore the runtime of the for-loop is 
$O(((n / k) \ln (n))\cdot n \log r) = O((n^{2} / k)\log n \log r)$
and the space is $O(n + \log r)$.

By \cref{lemma: rsalgo probably chooses a point in opt disc},
with probability at least
$1- n^{-1}$ we choose at least one 
of the $k$ points of $S$ from $D^*$ to be in our random sample
which, by \cref{lemma: bfdisc 2-approx}, means that with  
probability at least
$1- n^{-1}$
\ref{alg: RSAlgo} 
will output a 
$2$-approximation.
\end{proof}

\begin{remark-type-2}
There is a simple deterministic algorithm that runs in 
$O(n^2 \log r + m)$  time by computing the 
$(k-1)$\textsuperscript{st}-closest neighbour for each point.
However, we design a randomized algorithm, 
\ref{alg: RSAlgo},
that is faster when $k$ is $\omega(\log n)$.
Having a $k$ term in the denominator of the runtime
allows us to improve the runtime of our algorithm as $k$ approaches $n$.
When combining \ref{alg: RSAlgo} and \ref{alg: DI} 
(described in \cref{sec: Divide et Impera})
and ignoring polylogarithmic factors,
the runtime is at least a factor of $\Theta(n)$
faster than this simple approach. 
\end{remark-type-2}

It is worth pointing out that for $k \in \Omega (n)$, 
\textbf{\ref{alg: RSAlgo}} computes 
a $2$-SKEG disc with 
high probability
in $O(n\log n \log r + m)$ deterministic time 
using $O(n + m)$ space.

\section{Divide-and-Conquer Algorithm}
\label{sec: Divide et Impera}
In this section we describe 
\ref{alg: DI},
a Divide-and-Conquer
algorithm to compute a 
$2$-SKEG disc.
Omitted details appear in \cref{app: di-algo}.

Let $D^*$ be a SKEG disc for the 
points of $S$ in a polygon 
(the polygon being referred to will be clear from the context).
In each recursive call, we split the current polygon 
by a diagonal into two subpolygons of roughly equal size
and recursively compute a 
$2$-SKEG disc
for each of 
the two subpolygons. 
The merge step involves computing a $2$-approximation
to the optimal disc that contains $k$ points
under the assumption that
the optimal disc 
intersects the diagonal used to generate the recursive calls.
We delay discussion of
the merge step 
until \cref{section: quickselect merge step}.
\ref{alg: DI}
requires the following preprocessing
that 
constitutes \textbf{\ref{alg: DI preproc}}.
\ref{alg: DI preproc}
takes $O(n\log r + m)$ time
and $O( n + m)$ space.

\begin{description}
	\item[Simplification and Shortest-Path Data Structure] 
	Refer to \cref{sec: random sampling}.
	
	\item[Balanced Hierarchical Polygon Decomposition]
	We compute a balanced hierarchical polygon decomposition tree $T_B$
	\cite{DBLP:journals/dcg/ChazelleG89,GUIBAS1989126,GuibasHLST87}.
	(Note that $T_B$
    is built by the shortest-path
	data structure.)
  We continually insert diagonals into our polygon
	until all faces subdividing it are triangles;
	each polygonal face $f$ is split into two subpolygons 
	of 
	between $|f| / 3$ and $2|f| / 3$ 
	vertices
	by a diagonal of $f$.
	Decomposing our simplified polygon, this takes
	$O(r)$ time and $O(r)$ space
	\cite{chazelle1982theorem,GuibasHLST87}.
	We end up with a decomposition tree, based on a triangulation of our 
	polygon, of $O(\log r)$ height 
	whose leaves store the triangles.
	The internal nodes store the diagonals of the triangulation.
	A diagonal's position in the tree corresponds to when it was inserted.
    We associate with each internal node 
	the subpolygon split by the diagonal contained in the node
	(though we do not store this subpolygon in the node).
    The root is the initial polygon.
	
	\item[Point-Location Data Structure]
	Recall that the leaves of $T_B$ represent triangles
	in a triangulation of $P$.
	We build Kirkpatrick's \cite{Chazelle1991,Kirkpatrick83} 
	$O(\log r)$ query-time
	point-location data structure on these triangles 
	in $O(r)$ time with $O(r)$ 
	space.
	
	\item[Augment $T_B$ for Point Location]
	We augment $T_B$ to store which points of $S$
	are in which subpolygon represented by the internal nodes of the tree.
\end{description}

    \begin{restatable}{lemma}{pointsInSubPolLemma}
        \label{lemma: we know which points are in which subpolygon}
        In $O(n\log r + r)$ time and $O(n+r)$ space,
        we augment the decomposition tree $T_B$ to know which points of $S$
        are in which subpolygon.
    \end{restatable}

\begin{proof}
Recall that the leaves of 
$T_B$ represent triangles.
Using Kirkpatrick's \cite{Chazelle1991,Kirkpatrick83} $O(\log r)$ query-time
point-location data structure built on these triangles 
in $O(r)$ time with $O(r)$ 
space,
in $O(n\log r)$ time we then make a list of which points of $S$ are 
in each triangle at the base of the decomposition tree.
This can be done by storing a list in each 
triangle whose elements are the points in the triangle. 
When the point-location query returns a triangle for a given point of $S$, 
that point of $S$ is appended to the end of the triangle's list.
Since there are $n$ points of $S$, these lists take $O(n)$ space.
We then do a post-order traversal of $T_B$
in $O(n+r)$ time to determine how many
points of $S$ and which of them 
are in each subpolygon 
on either side of a given chord.\footnote{The two subpolygons being considered are those whose union results in the subpolygon split by the insertion of the chord (i.e., the subpolygon we associate with the diagonal's node in $T_B$).}
This can be done by creating an array $A$ of capacity $n$ and inserting
the points of $S$ into $A$ in the order that their triangle is
seen in the post-order traversal (i.e., the left-to-right order of
the leaves of 
$T_B$).
The time spent copying the points from the lists in the triangles is $O(n)$.
The leaves of $T_B$ mark where their interval in $A$ 
starts\footnote{They will know where their interval 
starts by the current size of $A$ when they copy their lists into $A$.}
and ends, 
and
the internal nodes have three pointers dictating where their left child's
interval starts and ends, and where their right child's interval starts
and ends.\footnote{We only need three pointers because the intervals are
consecutive.}
The interval of an internal node is from the beginning of its left 
child's interval to the end of its right child's.
This takes $O(n)$ space for $A$ and $O(1)$ space for each of the 
$O(r)$ nodes of the tree; therefore it takes $O(n + r)$ space
(the point location queries take $O(1)$ space).
\end{proof}	

\subsection{Algorithm Description: \ref{alg: DI}}
\label{subsec: DI alg descr}

Let $\tau$ denote the node of $T_B$ associated with the current recursive call,
let $P_{\tau}$ denote the current polygon split by the diagonal $\ell$ stored
in $\tau$, and let $S_{\tau}$ be the set of points of $S$ in $P_{\tau}$.
For ease of discussion, we abuse notation and say $P_{\tau}$ is stored in $\tau$.
Recall that the root of $T_B$ stores the initial polygon $P$.
Let $|P_{\tau}|=r'$ and $|S_{\tau}|=n'$.
Let $P_1$ (resp., $P_2$) be the 
subpolygon associated with the left (resp., right) child of $\tau$ on which we 
recurred
that contains the points $S_1 \subseteq S$ (resp., $S_2 \subseteq S$). 
We also use the notation $P_1 \cup P_2$ to refer to $P_{\tau}$.

The merge step (described in \cref{subsec: qs alg descr})
involves computing $D(c, \rho)$, a 
KEG disc 
for $S_1 \cup S_2$ 
where the centre $c \in \ell$. 
If $D^* \cap S_1 \neq \emptyset$
and $D^* \cap S_2 \neq \emptyset$ (see \cref{fig: merge disc a}),
then $D^*\cap \ell \neq \emptyset$ (as in \cref{fig: merge disc b}) and
either the specially-chosen centre $c$ returned from the merge step 
is inside $D^* \cap \ell$ 
(implying a $2$-approximation),
or $D(c, \rho)$ has a radius smaller 
than a disc centred on such a point
(also implying a $2$-approximation).
If 
either $D^* \cap S_1 = \emptyset$
or $D^* \cap S_2 = \emptyset$, 
then the optimal disc is centred in one of the subpolygons and contains
only points of $S$ in that subpolygon, thus we have a $2$-approximation by 
recursion.
Therefore, the smallest disc computed among the three discs
(i.e., the two discs returned by the
recursion and the disc computed in the merge step)
is a $2$-approximation to $D^*$.
For convenience, in the sequel we refer to 
the disc $D(c, \rho)$ whose centre is on the chord $\ell$
as a \emph{merge disc}.

\begin{figure}[t!]
\centering
	\begin{subfigure}{0.45\textwidth}
        \includegraphics[width=\linewidth,page=1]{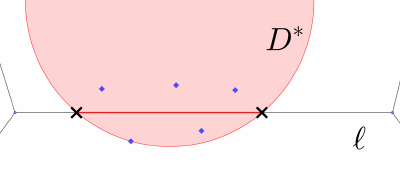}
		\caption{The red optimal disc $D^*$ contains points of $S_1$ 
		and $S_2$ (blue diamonds) from opposite sides of $\ell$.} \label{fig: merge disc a}
	\end{subfigure}
	\hspace*{\fill} 
	\begin{subfigure}{0.45\textwidth}
        \includegraphics[width=\linewidth,page=2]{figures/merge-larger-x.pdf}
		\caption{Centering a SKEG disc inside the red interval 
		$D^* \cap \ell$ (marked by two ``$\times$'' symbols) 
		produces a $2$-approximation.} \label{fig: merge disc b}
	\end{subfigure}
\caption{The case where the optimal disc $D^*$ contains points of $S$ from both sides of $\ell$.}
\label{fig: merge disc int ell}
\end{figure}

\textbf{\ref{alg: DI}} proceeds as follows.

\begin{enumerate}
	\item 
	(Base case) When the recursive step reaches a 
    triangle,\footnote{We could identify convex subpolygons before we get to the triangles, 
    but it would not improve the asymptotic runtime.}
	we use the planar 
	$2$-approximation algorithm 
	which runs in 
	expected time linear in
	the number of points of $S$ in the subpolygon under consideration
    \cite{har2011geometric,har2005fast}. 
	
	\item	
	We recur on the subpolygons
	$P_1$ and $P_2$ stored in the left and right child of $\tau$
	respectively.
	Note that a recursive call is not necessary if a subpolygon contains
	fewer than $k$ points of $S$ in it.
	Let the returned disc with the smaller radius be the current solution.
	
	\item (Merge Step)
	Consider the diagonal $\ell$ stored in $\tau$.
	Compute a merge disc centred on $\ell$ for the points of $S_1 \cup S_2$.
	
	\item
	Return the smallest of the three discs.

\end{enumerate}

\begin{restatable}{lemma}{mergeTwoApprox}
\label{lemma: merge is 2 approx}
\ref{alg: DI} produces a $2$-SKEG disc for $P$.
\end{restatable}

\begin{proof}
We proceed by induction on the number of vertices $r$ of the polygon. 
\begin{description}
    \item[Base Case] In the base case, we have $r = 3$ for a triangle. 
    We get a $2$-approximation by running the 
    planar $2$-approximation algorithm \cite{har2005fast}.

    \item[Inductive Hypothesis] 
    Assume the Lemma holds for polygons with at most $t$ vertices.
    
    \item[Inductive Step] 
    Let $|P_{\tau}| = t+1$.
    Partition $P_{\tau}$ into two simple subpolygons $P_1$ and $P_2$
    by the diagonal stored in $\tau$
    containing the points $S_1\subseteq S$ and $S_2\subseteq S$ respectively.
    Note that $|P_1|$ and $|P_2|$ are both at most $t$.
    By the inductive hypothesis, we have computed a $2$-SKEG disc
    for $P_1$ and $S_1$, and for $P_2$ and $S_2$.
    
    Consider the merge disc returned from our merge step.
    We show in \cref{section: quickselect merge step} that
    if
    both $D^* \cap S_1 \neq \emptyset$ and $D^* \cap S_2 \neq \emptyset$
    then the merge disc is a $2$-approximation. 
    Indeed, if both $D^* \cap S_1 \neq \emptyset$ and 
    $D^* \cap S_2 \neq \emptyset$,
    then $D^* \cap \ell \neq \emptyset$.
    In \cref{lemma: something close in opt interval} in \cref{section: quickselect merge step} 
    we prove that at least one point in $D^* \cap \ell$ is considered
    to be the centre of the merge disc.
    As such, 
    either the centre of the merge disc 
    is inside $D^* \cap \ell$ 
    (implying a $2$-approximation), 
    or 
    the radius of the merge disc is smaller than 
	that of a disc centred on
	at least one point in
	$D^* \cap \ell$
    (also implying a $2$-approximation).
    If either $D^* \cap S_1 = \emptyset$ or $D^* \cap S_2 = \emptyset$, 
    we have 
    a $2$-approximation by the inductive hypothesis 
    and the fact that in this case
    the optimal disc for $S_1 \cup S_2$ is the same as for whichever 
    set has a non-empty
    intersection with $D^*$.
    Of these three discs, the one with the smallest radius is 
    a $2$-approximation
    for $S_1 \cup S_2$.
\end{description}
\end{proof}

\begin{restatable}{theorem}{diAnalysis}
\label{cor: DC using quickselect}
\textbf{\ref{alg: DI}} computes a
$2$-SKEG disc
in 
$O(n \log^2 n \log r + m)$ expected time and 
$O(n + m)$ expected space.
\end{restatable}

\begin{proof}
As discussed earlier (at the beginning of \cref{sec: Divide et Impera}), 
preprocessing takes $O(n\log r + m)$ time and 
$O(n + m)$ space.

\ref{alg: DI}
is a recursive, Divide-and-Conquer algorithm.
The recurrence tree of the Divide-and-Conquer algorithm
mimics 
$T_B$
and has $O(\log r)$ depth.
The recursive algorithm visits each node of the tree
that represents a subpolygon that contains at least $k$ points of $S$.
For nodes representing subpolygons containing fewer than $k$ points of $S$,
there is no work to be done; such a node and the branch of $T_B$ 
stemming from it are effectively 
pruned from the recursion tree.
By \cref{lemma: merge is 2 approx}, the result of the merge step
is a $2$-approximation for $P_{\tau}$.
Therefore, when we finish at the root, we have a $2$-approximation 
to $D^*$.

The base case of the recursion is triggered when we reach
a leaf of $T_B$.
Here 
$P_{\tau}$
is a triangle.
When the base case is reached,
\ref{alg: DI} 
runs the planar $2$-approximation algorithm 
in expected time and expected space 
linear in the number of points of 
$S_{\tau}$ (i.e., $O(n')$)
\cite{har2011geometric,har2005fast}.

We assume for the moment that the merge step
runs in expected time
$O(n' \log n' \log r + n' \log^2 n')$ 
and uses $O(n' + r')$ space
(we show this in \cref{subsec: qs alg descr}).
The expected running time of the base case is
dominated by the time for the merge step.
This implies that the algorithm runs in 
$O(n \log n \log^2 r + n \log^2 n \log r + m)$ expected time,
including the time for preprocessing.
We can simplify this expression to $O(n\log^2 n\log r + m)$.
The space bound follows from the space for preprocessing
and the space in the merge step 
which is released after the merge.
\end{proof}

We can balance this expected runtime 
against the runtime of \cref{theorem: rs algo}.

\begin{align*}
&& n \log^2 n \log r & = \frac{n^2}{k}\log n \log r & \\
\Rightarrow && n \log n & = \frac{n^2}{k} & \\
\Rightarrow && k & = \frac{n}{\log n} &
\end{align*}

When $k \in O(n / \log n)$,
the expected 
runtime
of 
\ref{alg: DI} is faster
by \cref{cor: DC using quickselect};
when $k \in \omega(n / \log n)$ 
the 
runtime
of 
\ref{alg: RSAlgo}
is
faster by polylogarithmic factors 
by \cref{theorem: rs algo}.
Although the runtimes are asymptotically identical
when $k \in \Theta(n/\log n)$, 
\ref{alg: RSAlgo}
gets faster as $k$ gets larger.
For example, if $k \in \Omega (n)$ 
\ref{alg: RSAlgo} 
runs in
time $O(n \log n \log r + m)$ and finds a $2$-approximation with high probability,
whereas \cref{cor: DC using quickselect} finds a $2$-approximation in expected time 
$O(n \log^2 n \log r + m)$.

This leads to our main theorem and our main algorithm, \textbf{\ref{alg: main}}.
\ref{alg: main} first performs the preprocessing of \ref{alg: DI preproc}.
Then, if the polygon is convex, it runs the linear-time planar 
approximation algorithm \cite{har2011geometric,har2005fast}.
If the polygon is not convex and $k \in \omega(n/\log n)$, it runs \ref{alg: RSAlgo}.
Otherwise, it runs \ref{alg: DI}.

{\renewcommand\footnote[1]{}\NewMainResultTheorem}

\section{Merge}
\label{section: quickselect merge step}
In this section we describe how to perform
the merge step of our Divide-and-Conquer algorithm.
First we point out that it is not clear whether 
it is possible to apply the recursive
random sampling technique of Chan \cite{DBLP:journals/dcg/Chan99}.
His approach requires 
partitioning
the points of $S$ into a constant
number of fractional-sized
subsets such that the overall solution is the best of the solutions
of each of the subsets.
It is not clear how to partition the points of $S$ 
to allow for an efficient merge step.
In fact, this is an issue that Chan \cite{DBLP:journals/dcg/Chan99} points out in his paper 
for a related problem of finding the smallest square containing $k$ points
before showing how to circumvent this issue in an orthogonal setting.

Assume a SKEG disc in $P_{\tau}$ for $k$ points of $S_{\tau}$ intersects $\ell$
and contains at least one point of $S_1$ and one point of $S_2$.
In this case our merge step either returns:
$1)$ a KEG disc centred on a special point of $\ell$
that is
guaranteed to be inside a SKEG disc;
or $2)$ a KEG disc centred on $\ell$ whose radius is smaller than 
that of such a KEG disc as described in $1)$.

\subsection{Merge Algorithm}
\label{subsec: qs alg descr}
For each $u \in S_{\tau}$, 
let $u_c$ be the closest point of the chord $\ell$ to $u$.
Let the set of all such closest points be $S^{c}_{\tau}$.
We refer to the elements of $S^{c}_{\tau}$ as \emph{projections} 
of the elements of 
$S_{\tau}$ 
onto $\ell$ (see \cref{fig: projections}).
For any $u \in S_{\tau}$ 
and radius $\rho$, let the interval $I(u,\rho) \subseteq \ell$ be $D(u, \rho) \cap \ell$ 
(i.e., the set of points on $\ell$ within geodesic distance $\rho$ of $u$).
See \cref{fig: discs int ell a}.
$I(u,\rho)$ is empty if $\rho < d_g(u, u_c)$.
For $u_c \in S^{c}_{\tau}$ 
and radius $\rho$,
we say the \emph{depth} of $u_c$ is the number of 
intervals from points of 
$S_{\tau}$
that contain $u_c$ using the same
distance $\rho$ 
(i.e., the number
of discs of radius $\rho$ centred on points of 
$S_{\tau}$
that contain $u_c$).

\begin{figure}[t!]
		\includegraphics[width=\linewidth]{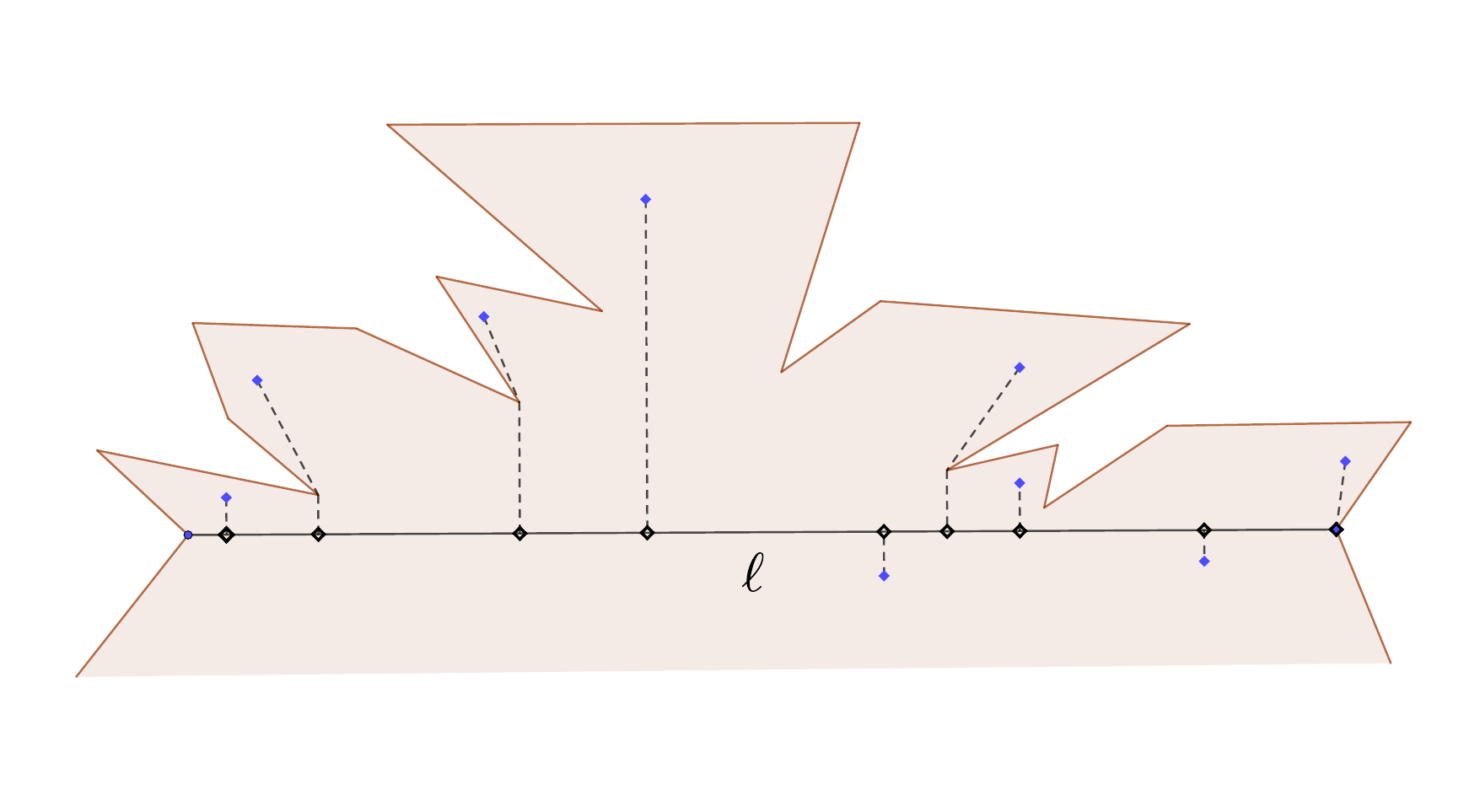}
		\caption{An example of points of $S$ (blue diamonds) and their projections onto $\ell$ (hollow black diamonds).} \label{fig: projections}
\end{figure}

\begin{figure}[t!]
		\includegraphics[width=\linewidth]{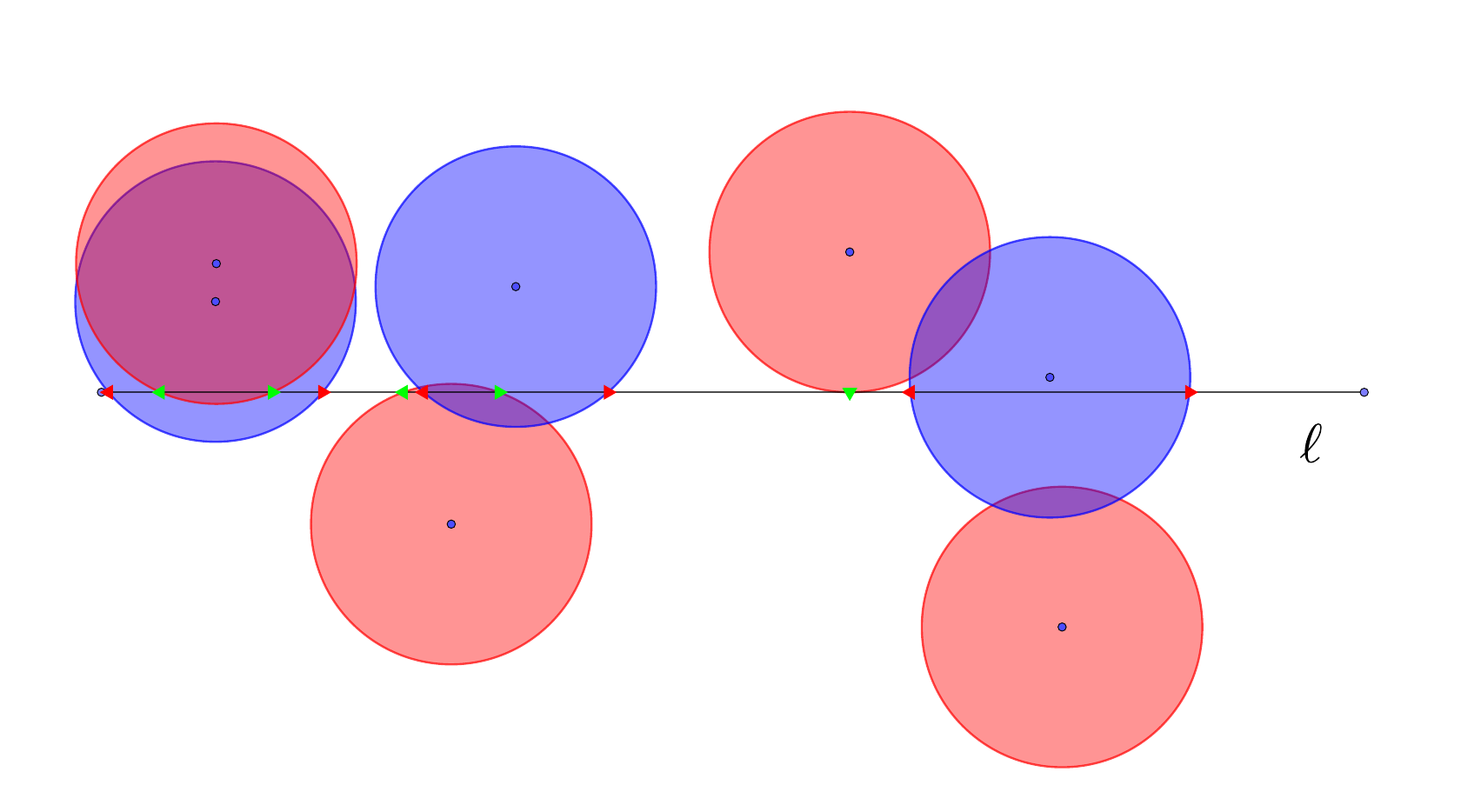}
		\caption{The geodesic discs (arbitrarily red and blue) of 
		radius $\rho$ centred on points of $S$ (blue points) 
		intersect $\ell$.
		The intersection points of red (blue) disc boundaries 
		with $\ell$ are 
		marked by green (red) triangles.
  The intervals $I(u,\rho)$ for the points $u \in S$ are the 
  intersections of $\ell$ with $D(u, \rho)$.
  Overlapping intervals illustrate points along $\ell$ where 
  centering a geodesic of radius $\rho$ will contain multiple 
  points of $S$ (i.e., the points of $S$ defining the overlapping intervals).} \label{fig: discs int ell a}
\end{figure}

\begin{observation}
If the SKEG disc $D^*$ contains at least one point of $S_1$ and at least one point of $S_2$,
then $D^* \cap \ell$ is a continuous non-empty interval of $\ell$ due to geodesic convexity.
\end{observation}

\textbf{\ref{alg: merge-me}}
proceeds as follows.

\begin{enumerate}
	\item 
	Compute $\mathbb{C} = S^{c}_{\tau}$.
	
	\item
    Initialize set $\mathbb{S} = S_{\tau}$.
	
	\item	
	While $|\mathbb{C}| > 0$:
	\begin{itemize}
	    \item 
	    Pick a point $u_c \in \mathbb{C}$ uniformly at random.
	    
	    \item
	    Find the point $z \in \mathbb{S}$ that is the $k$\textsuperscript{th}-closest neighbour of $u_c$.
	    
	    \item
	    Let $\rho = d_g(u_c, z)$.
	    
	    \item
	    If $|\mathbb{C}| = 1$, break the while-loop.
	    
	    \item
	    For each $v \in \mathbb{S}$, compute $I(v, \rho)$. 
	    If  $I(v, \rho)$ is empty, remove $v$ from $\mathbb{S}$.
	    
	    \item 
	    For each $w \in \mathbb{C}$, compute the depth of $w$.
	    If the depth of $w$ is less than $k$, remove $w$ from $\mathbb{C}$.
	    
	    \item
	    Remove $u_c$ from $\mathbb{C}$.
	\end{itemize}
	
	\item 
	Return $u_c$ and $\rho$.
\end{enumerate}

\ref{alg: merge-me}
computes
the point $u^{*}_{c} \in S^{c}_{\tau}$
whose distance to its $k$\textsuperscript{th}-nearest 
neighbour in $S_{\tau}$ is minimal.
In essence, we are given a set $\Gamma$
of $n$ numbers (i.e., the distances to the 
$k$\textsuperscript{th}-nearest neighbours)
and want to find the minimum.
At each step, the algorithm picks a random element
$g \in \Gamma$ and removes all elements in $\Gamma$ larger
than $g$.
When we prove the correctness of 
\ref{alg: merge-me}
we will 
show that the
number of iterations is 
$O(\log |\Gamma|)$ with high probability.

\begin{restatable}{lemma}{somethingCloseInOptInterval}
\label{lemma: something close in opt interval}
If a SKEG disc $D^*$ contains at least one point of 
$S_{\tau}$ 
from each side of $\ell$,
then for some 
$u \in S_{\tau}$ 
that is in $D^*$, $D^* \cap \ell$ contains 
$u_c \in S^{c}_{\tau}$.
\end{restatable}

\begin{proof}
Without loss of generality, assume $u \in D^*$ is below $\ell$ and
the centre $c^*$ of $D^*$ is 
on or above $\ell$.
Assume that the chord $\ell$ is horizontal for ease of reference.
It was shown by Pollack et al.\ \cite[Corollary $2$]{DBLP:journals/dcg/PollackSR89} that,
for any triplet of points $p$, $v$, and $w$ in a simple polygon,
if the first edge of $\Pi(p,v)$ makes an angle of $\pi/2$ or 
greater with the first edge of $\Pi(p,w)$,
then $d_g(v, w)$ is larger than both
$d_g(p, v)$ and $d_g(p,w)$.
We use this  
fact to complete the proof.

Consider $\Pi(c^*, u)$.
Since $c^*$ and $u$ are on different sides of $\ell$,
$\Pi(c^*, u)$ must intersect $\ell$ at some point $z$.
We want to show that $u_c \in D^*$.
If $z = u_c$, then $\Pi(c^*, z) = \Pi(c^*, u_c) \subseteq \Pi(c^*, u)$
and thus $u_c \in D^*$.
If $z \neq u_c$, then let us consider the geodesic triangle
formed by the union of $\Pi(u_c, z)$, $\Pi(u_c, u)$, and $\Pi(z, u)$.
Notice that if $u_c$ is not an endpoint of $\ell$ then
the angle at $u_c$ (i.e., the angle formed by the first edges of
$\Pi(u_c, z)$ and $\Pi(u_c, u)$) is $\pi/2$ since $\Pi(u_c, z) \subseteq \ell$;
and by the Pythagorean theorem the first edge on $\Pi(u_c, u)$
forms an angle of $\pi/2$ with $\ell$ (since $u_c$ is the closest point of $\ell$ to $u$).
This means that $|\Pi(z, u)| > |\Pi(u_c, z)|$
(by Pollack et al.\ \cite[Corollary $2$]{DBLP:journals/dcg/PollackSR89}).
Thus, in this case, $u_c \in D^*$ by the triangle inequality:
$|\Pi(c^*, u_c)| \leq |\Pi(c^*, z)| + |\Pi(z, u_c)| < |\Pi(c^*, z)| + |\Pi(z, u)| = |\Pi(c^*, u)|$.
If $z \neq u_c$ and $u_c$ is an endpoint of $\ell$, 
then the angle at $u_c$ is at least $\pi/2$ and the same argument applies.
\end{proof}

\begin{restatable}{lemma}{mergeRuntime}
\label{lemma: simple rand works}
\textbf{\ref{alg: merge-me}} 
runs
in $O(n' \log n' \log r + n' \log^2 n')$ 
time with high probability 
and $O(n'+r')$ space and produces a $2$-approximation
if $D^* \cap S_1 \neq \emptyset$ and $D^* \cap S_2 \neq \emptyset$.  
\end{restatable}

\begin{proof}
If $D^* \cap S_1 \neq \emptyset$ and $D^* \cap S_2 \neq \emptyset$, 
then we either return a disc centred
on a projection of some point of 
$S_{\tau}$
inside $D^*$, or a disc centred on a projection
whose radius is less than that of such a disc.
Thus, the result of the algorithm is a $2$-approximation 
by \cref{lemma: something close in opt interval}. 

Recall that $|S_{\tau}| = n'$ and $|P_{\tau}| = r'$.
We assume for the moment that we can
compute 
$S^{c}_{\tau}$
in $O(n'\log r)$ time 
and $O(n'+r')$ space (we will show this in 
\cref{subsubsec: testing closest points}).
Given the $n'$ elements of 
$S^{c}_{\tau}$ or $S_{\tau}$:
we build sets $\mathbb{C}$ and $\mathbb{S}$ in $O(n')$ time and space;
we can delete an identified element from the sets
in constant time and space;
we can iterate through a set
in constant time and space per element;
and we can pick a point $u_c \in \mathbb{C}$ uniformly at random 
in $O(1)$ time and space.

Using the shortest-path data structure, we can find the 
$k$\textsuperscript{th}-closest 
neighbour of $u_c$ in $O(n'\log r)$ time and $O(n' + r')$ space by first
computing the distance of everyone in $\mathbb{S}$ to 
$u_c$ in $O(\log r)$ time
and $O(r')$ space each (the space is re-used for the next query), 
and then using a linear-time rank-finding algorithm
to find the $k$\textsuperscript{th}-ranking 
distance in $O(n')$ time and space
\cite{DBLP:conf/wea/Alexandrescu17,DBLP:journals/jcss/BlumFPRT73,DBLP:books/daglib/0023376}.
Let $\rho$ be this $k$\textsuperscript{th}-ranking distance.
We assume for the moment that we can compute,
for all $O(n')$ elements $v \in \mathbb{S}$,
the intervals $I(v, \rho)$ in 
$O(n'\log r)$ time
and 
$O(n'+r')$ space 
(we will show this in \cref{subsec: decision prob par s}).
Given the intervals, we also assume for the moment that we can compute,
for all $O(n')$ elements $w \in \mathbb{C}$, the depth of $w$
in overall $O((|\mathbb{S}| + |\mathbb{C}|)\log (|\mathbb{S}| + |\mathbb{C}|)) = O(n' \log n')$ time 
and $O(|\mathbb{S}| + |\mathbb{C}|) = O(n')$ space
(we will show this in \cref{subsec: decision prob par s}).
In 
$O(n')$ 
time and space we can remove from $\mathbb{S}$
and $\mathbb{C}$ elements whose interval $I$ is empty or whose
depth is less than $k$, respectively.
Elements of $\mathbb{S}$ whose intervals are empty are too far away from
$\ell$ to be contained in the geodesic disc of any element of $\mathbb{C}$ for 
any radius of value $\rho$ or less, and since $\rho$ is non-increasing in each
iteration, we can remove these elements of  $\mathbb{S}$ from consideration.\footnote{At the moment, 
removing points from  $\mathbb{S}$ is a practical consideration; no lower bound is clear on
the number of elements from $\mathbb{S}$ that can be discarded in each iteration, though
if one were to find a constant fraction lower bound, one could shave a logarithmic factor off the runtime.}
Elements of $\mathbb{C}$ whose depth is less than $k$ will contain fewer than $k$ elements of $\mathbb{S}$
in their geodesic discs for any radius of value $\rho$ or less.
Since we seek to minimize the value $\rho$, 
we can remove these elements of $\mathbb{C}$ from consideration.
The removal of these elements
ensures $\rho$ is non-increasing in each iteration.
Computing the $k$\textsuperscript{th}-nearest neighbour of a projection,
the intervals $I$, 
and the depths of elements of $\mathbb{C}$ 
dominates the complexity of the while-loop.

We now analyze how many iterations of the while-loop are performed.
We associate with each projection 
in $\mathbb{C}$ 
the distance associated with its 
$k$\textsuperscript{th}-nearest neighbour 
(and thus the disc's radius) and assume 
these distances are unique.
Consider a random permutation of these distances.
This permutation can be seen as the order in which
the elements are inserted into a random binary search tree.
Devroye \cite{DBLP:journals/acta/Devroye88}
pointed out that the depth (i.e., the level in the binary search tree)
of the next item to be inserted
is proportional to the sum of: 
the number of up-records seen so far whose
value is less than the item to be inserted
(i.e., the number of elements in the ordered sequence so far that
were bigger than everything that came before them, 
but still less than the item to be inserted);
and of the number of down-records seen so far
whose value is larger than the item to be inserted 
(i.e., the number of elements in the ordered sequence so far that
were smaller than everything that came before them, 
but still more than the item to be inserted).
Every node visited on the path from the root to this item
in the constructed tree is one of these records.
Devroye \cite{DBLP:journals/acta/Devroye88} then
showed that with high probability 
the depth of the last element 
is $O(\log n')$ for a sequence of $n'$ items.
In our while-loop,
we are searching for the element of smallest rank.
This can be seen as following the path from the root of
this random binary search tree to the smallest element.
In each iteration, we pick 
an element of $\mathbb{C}$ uniformly at random.
The radius defined by this element will always be a down-record
since it is always $k$-deep (since it defines the radius),
and since the next step in the iteration is to remove
elements whose $k$\textsuperscript{th}-nearest neighbour
is farther than the chosen radius.
Each time we pick an element at random, it is like encountering
a down-record in the random permutation of distances
that can be used to create the random binary search tree.
The elements that are discarded in an iteration are like the
elements in the sequence between two down-records.
Thus, with high probability, we require $O(\log n')$ iterations.

The work done in one iteration given sets $\mathbb{S}$ and $\mathbb{C}$ 
with $|\mathbb{S}| = n'$ and $|\mathbb{C}| \leq n'$ is $O(n'\log r + n' \log n')$ using $O(n'+r')$ space.
With high probability we have $O(\log n')$ iterations,
thus $O(n'\log n' \log r + n' \log^2 n')$ time with high probability.
\end{proof}

\subsection{With Regards to Shortest Paths}
\label{app: discussing shortest paths and distance functions}
The shortest-path data structure of Guibas and Hershberger
\cite{GUIBAS1989126,HERSHBERGER1991231} represents the 
shortest path between two points in $P$ as a tree
of $O(\log r)$ height whose in-order traversal reveals
the edges of the shortest path in order. 
A query between any source and destination points in $P$ 
is performed in $O(\log r)$ time and additional space
\cite{GUIBAS1989126,HERSHBERGER1991231}.
The data structure stores the distance from a vertex of the 
shortest-path to the source in the query, so given a vertex of the path
the distance to the source can be reported in $O(1)$ time and space.

\begin{description}
	\item[Funnel] Now let us briefly review the notion of a \emph{funnel} 
	\cite{GUIBAS1989126,DBLP:journals/networks/LeeP84}.
	The vertices of the geodesic shortest path between two points 
	$a$ and $b$, $\Pi (a,b)$, consists of the vertices $a$, $b$, and a subset
	of the vertices of the polygon $P$ forming a polygonal chain 
	\cite{DBLP:journals/networks/CheinS83,DBLP:journals/cacm/Lozano-PerezW79}.
	Consider a diagonal $\ell$ of $P$, its two endpoints $\ell_1$ and $\ell_2$,
	and a point $p \in \{P \cup S\}$.
	The union of the three paths $\Pi(p, \ell_1)$, $\Pi(p, \ell_2)$, and $\ell$ 
	form what is called a funnel.
	As Guibas and Hershberger \cite{GUIBAS1989126} point out, this funnel
	represents the shortest paths from $p$ to the points on $\ell$ in that
	their union is the funnel.
	Starting at $p$, the paths $\Pi(p, \ell_1)$ and $\Pi(p,\ell_2)$ may overlap
	during a subpath, but there is a unique vertex $p_a$ (which is the farthest
	vertex on their common subpath from $p$) where the two paths diverge.
	After they diverge, the two paths never meet again.
	This vertex $p_a$ is called the \emph{apex}\footnote{Sometimes, as in 
	\cite{DBLP:journals/networks/LeeP84}, this is called a \emph{cusp}.
	In \cite{DBLP:journals/networks/LeeP84} the funnel is defined as beginning at the
	cusp and ending at the diagonal.} of the funnel.
	The path from the apex to an endpoint of $\ell$ forms an
	\emph{inward-convex} polygonal chain (i.e., a convex path through vertices of $P$
	with the bend protruding into the interior of $P$).
\end{description}

Below we present bounds for operations in a subpolygon.
We modify our perspective slightly by
re-defining the diagonal to which we refer
to make the discussion simpler.
Let $\tau_{up}$ be the parent of $\tau$ in $T_B$.
We change the definition of $\ell$ to be the diagonal stored in $\tau_{up}$.
Note that  $\ell \in \partial P_{\tau}$ (which is a diagonal of $P$).
We consider $\ell$ to be the $x$-axis with
the downward direction the side of $\ell$ bordering the exterior of
$P_{\tau}$;
$\ell_1$ to be the left endpoint of $\ell$;
and $\ell_2$ to be
the right endpoint of $\ell$.
Without loss of generality, assume the points of $S_{\tau}$ are \emph{above} $\ell$.
Recall that $|P_{\tau}|= r'$ and $|S_{\tau}| = n'$.
A careful reading of Guibas and Hershberger \cite{GUIBAS1989126,HERSHBERGER1991231}
gives us the following. 

\begin{observation}
\label{lemma: sp query size}
The returned 
tree for a shortest-path query for two points in 
$P_{\tau}$
has $O(r')$ nodes and height $O(\log r')$.
\end{observation}

Guibas and Hershberger \cite{GUIBAS1989126}
and Oh and Ahn \cite{Oh2019} point out that given the trees
representing the shortest paths between a fixed source and two 
distinct destination points of a chord, 
the apex of their funnel can be
computed in $O(\log r)$ time.

\begin{observation}
\label{lemma: apex in log time}
The apex of a funnel 
from a source point in 
$P_{\tau}$
to the diagonal $\ell$ 
can be computed in
$O(\log r)$ time and $O( r')$ space.
\end{observation}

\begin{definition}[{Aronov $1989$ \cite[Definition $3.1$]{Aronov89}}]
For any two points $u$ and $v$ of $P$, 
the last vertex (or $u$ if there is none) before $v$ on 
$\Pi(u, v)$ is referred to as the \textbf{anchor} of $v$ 
(with respect to $u$).
\end{definition}

\begin{definition}[{Aronov $1989$ \cite[Definition $3.7$]{Aronov89}}]
The \emph{shortest-path tree} of $P$ from a point $s$ of $P$, 
$T(P,s)$, is the
union of the geodesic shortest paths from $s$ to vertices of $P$.
This type of tree, which is a union of paths in the polygon,
should not be confused with the type of tree returned from a 
query to the shortest path data structure, 
which is a tree connecting data-structure nodes.
\end{definition}

\begin{definition}[{Aronov $1989$ \cite[paragraph between $3.8$ and $3.9$]{Aronov89}}]
Let $e$ be an edge of $T(P,s)$ and let its endpoint furthest from
$s$ be $v$. 
Let $h$ be the open half-line collinear with $e$ and 
extending from $v$ in the direction of increasing distance from $s$.
If some initial section of $h$ is contained in the interior of $P$, we 
will refer to the maximal such initial section as the 
\emph{\textbf{extension segment}} of $e$.
\end{definition}

\begin{definition}[{Aronov $1989$ \cite[Definition $3.9$]{Aronov89}}]
Let the collection of extension segments of edges of $T(P,s)$ be denoted
by $E(P,s)$.
\end{definition}

\begin{description}
	\item[Distance Function of a Point $\bm{u \in S_{\tau}}$] Next, let us review 
	the graph we get by
	plotting the distance from a point $u$ to a line $\ell$
	where the position along $\ell$ is parameterized by $x$.
	Abusing notation, we call the $x$-monotone 
	curve representing this graph 
	the \emph{\textbf{distance function}}.
	Without loss of generality, we can assume that the $x$-axis is the line
	in question.
	For a point $u$, $u_x$ is the $x$-coordinate of $u$ and $u_y$ is its 
	$y$-coordinate.
	This distance function is actually a branch of a
	\emph{right hyperbola}\footnote{Also called a \emph{rectangular} or
	\emph{equilateral} hyperbola.} 
	whose eccentricity is $\sqrt{2}$ and whose focus
	is therefore at $\sqrt{2}\cdot u_y$.
	In our polygon 
    $P_{\tau}$,
	this distance function measuring the distance
	from a point 
    $u \in S_{\tau}$
	to a chord $\ell$ of the polygon becomes a 
	continuous piecewise hyperbolic function.
	If the funnel from $u$ to the endpoints $\ell_1$ and $\ell_2$ 
	of $\ell$ is trivial
	(i.e., a Euclidean triangle), 
	then the distance function of $u$ to $\ell$ has one piece 
	expressed as $\operatorname{dist}_u(x)=\sqrt{(x - u_x)^2 + u_y^2}$.
	We write $\operatorname{dist}_u(\cdot)$ to refer to the function
	in general.
	If, on the other hand, there are reflex vertices of $P$  
	in $u$'s funnel, the distance function has multiple pieces.
	The formula for each piece is
	$\operatorname{dist}_u(x)=\sqrt{(x - w_x)^2 + w_y^2} + d_g(u,w)$, where $w$ is the 
	anchor (i.e., the last vertex of $P$ on $\Pi(u,x)$), and
	the \emph{\textbf{domain}} of this hyperbolic piece is the set of values of $x$ for which $w$ is 
	the anchor. 
    The domain
    is defined by the intersection of $\ell$ with the 
	two extension segments of 
    $E(P_{\tau}, u)$
	that go through $w$, i.e., without loss of generality, 
	if $w$ is on the path from $u$ to $\ell_1$ but is not the apex of the funnel,
	then we are referring to the extension segment through $w$ and the predecessor
	of $w$ on $\Pi(u, \ell_1)$, and the extension segment through $w$ and the
	successor of $w$ on $\Pi(u, \ell_1)$.
	The domain of the apex is similarly defined, except there will be one extension
	segment through the apex and its successor in $\Pi(u, \ell_1)$, and
	one through the apex and its successor in $\Pi(u, \ell_2)$.	
    Given a vertex of the funnel, 
    we can build the hyperbolic piece associated with the vertex
    in $O(1)$ time and space
    since the shortest-path query data structure of 
    Guibas and Hershberger \cite{GUIBAS1989126}
    stores the distance from the source.
	Refer to Fig.~\ref{fig: dist function fig} for an example of a multi-piece
	distance function.
\end{description}	

\begin{figure}
\includegraphics[page=1,scale=0.5]{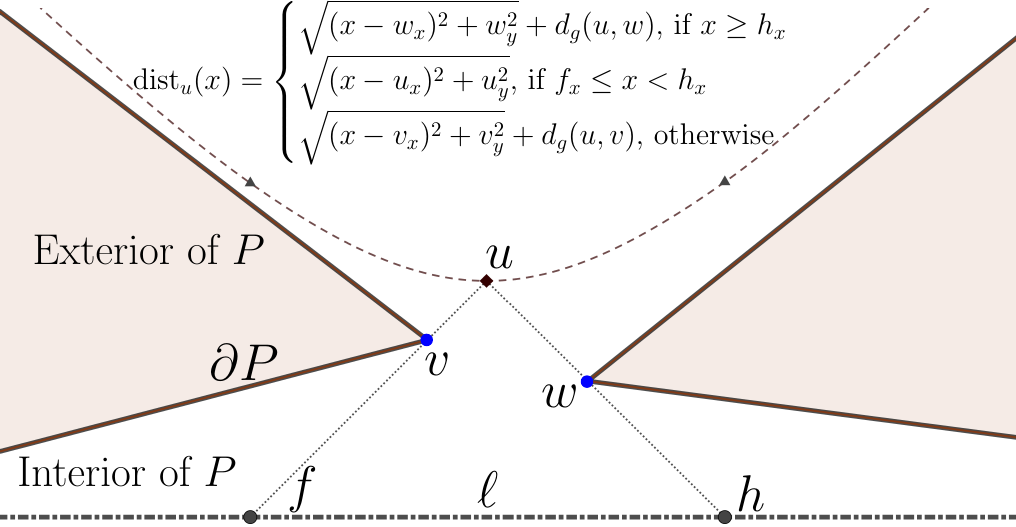}
\centering
\caption{Considering the chord $\ell$ of $P$ to be the $x$-axis,
given a point $u \in S$
we refer to the dashed 
graph of the function
$\operatorname{dist}_u(\cdot)$ as the distance function
of $u$ to $\ell$.
The points $f$ and $h$ on $\ell$ mark where different pieces of 
$\operatorname{dist}_u(\cdot)$ begin. 
}
\label{fig: dist function fig}
\end{figure}

Consider a point 
$u \in S_{\tau}$
and the subset 
$E \subseteq E(P_{\tau}, u)$
whose elements form
the domains of the pieces of $\operatorname{dist}_u(\cdot)$ along $\ell$.
For a given $e \in E$, we refer to $e \cap \ell$ as a \emph{\textbf{domain marker}} 
(or just \emph{\textbf{marker}}).
Sometimes we will need to identify domains
that have specific properties so that an appropriate hyperbolic piece of
$\operatorname{dist}_u(\cdot)$ can be analyzed.
Similar to other papers that find intervals of interest along 
shortest paths and chords
\cite{DBLP:conf/compgeom/AgarwalAS18,DBLP:journals/corr/abs-1803-05765,DBLP:journals/corr/ArgeS17,Oh2019},
we can use the funnel between $u$ and $\ell$ to perform a binary search among the domain markers
to find a domain of interest.
We have the following observation.

\begin{restatable}{observation}{intervalSearchLemma}
\label{lemma: logr time interval search}
	For an extension segment $e \in E$, 
 	if it takes $O(1)$ time and space to 
	determine which side of $\ell \cap e$
	contains a domain of interest along $\ell$,
	then we can find a domain of 
	interest along $\ell$ and its corresponding hyperbolic piece
	of $\operatorname{dist}_u(\cdot)$
 	in $O(\log r )$ time
 	and $O(r')$ space.
\end{restatable}

\subsection{Computing Projections}
\label{subsubsec: testing closest points}
After briefly re-defining $\ell$ in 
\cref{app: discussing shortest paths and distance functions},
we now re-adjust our perspective by resetting $\ell$ to be the
diagonal stored in $\tau$ that splits $P_{\tau}$.

Before the main while-loop of 
\ref{alg: merge-me} 
begins,
we precompute 
$S^{c}_{\tau}$, 
i.e., 
the closest point of $\ell$ to $u$
for each point 
$u \in S_{\tau}$.
Recall that $|P_{\tau}|= r'$ and $|S_{\tau}| = n'$.
As before, let $u_c$ be the closest point of $\ell$ to $u$.
Equivalently, $u_c$ is the point along
$\ell$ that minimizes $\operatorname{dist}_u(\cdot)$.

\begin{restatable}{lemma}{logrClosestPoint}
\label{lemma: logr closest point}
In $O(n'\log r)$ time and $O(n' + r')$ space we compute the closest point of $\ell$ to each 
$u \in S_{\tau}$.
\end{restatable}

\begin{proof}
Consider a point $p \in \ell$ and the last edge $e$ of $\Pi(u,p)$ from 
some 
$u \in S_{\tau}$
to $p$ (i.e., the edge to which $p$ is incident).
Let the \emph{angle of $e$} be the smaller of the two angles
formed by $e$ and $\ell$ at $p$. 
The range of this angle is $[0, \pi / 2]$.
We know from Pollack et al.\ \cite[Corollary $2$]{DBLP:journals/dcg/PollackSR89} that
$\operatorname{dist}_u(\cdot)$ is minimized when $e$ is perpendicular to $\ell$.
We also know from Pollack et al.\ \cite[Corollary $2$]{DBLP:journals/dcg/PollackSR89} 
that given $p' \in \ell$ and an edge $e'$ analogous to $e$, if the angle 
of $e'$ is closer to $\pi/2$ than that of $e$, then 
$\operatorname{dist}_u(p') < \operatorname{dist}_u(p)$.
Lastly, we know from 
Pollack et al.\ \cite[Lemma $1$]{DBLP:journals/dcg/PollackSR89}  
that $\operatorname{dist}_u(\cdot)$ is a convex function which means 
it has a global minimum.

Recall the notions of 
\emph{domains} of funnel vertices along $\ell$
and
\emph{markers} of extension segments collinear with funnel edges.
Given the funnel of $u$ and $\ell$,
by \cref{lemma: apex in log time} we can find 
the apex $u_a$ of the funnel  in $O(\log r)$ time
and $O(r')$ space.
The convex chains $\Pi(u_a, \ell_1)$ and $\Pi(u_a, \ell_2)$ determine
the hyperbolic pieces of $\operatorname{dist}_u(\cdot)$.
Given $u_a$, 
a careful reading of Guibas and Hershberger \cite{GUIBAS1989126,HERSHBERGER1991231} shows
we can compute the representations of $\Pi(u_a, \ell_1)$ and $\Pi(u_a, \ell_2)$
with shortest-path queries
in $O(\log r')$ time and $O(r')$ space.
By \cref{lemma: logr time interval search} 
we can use these convex chains to perform a binary search along $\ell$ to find the 
domain in which $u_c$ lies (see \cref{fig: funnel-search}).
This domain has the property that
the angle of the last edge on the shortest path 
from $u$ to the points in this domain
is closest to $\pi / 2$.

In the binary search,
at each 
marker
(as determined by the node currently being visited in the tree representing the convex chain), 
in $O(1)$ time and space we compute the angle of 
$\ell$ with the extension segment defining the marker.
Since $\operatorname{dist}_u(\cdot)$ is a convex function, 
we know that as we slide a point $p \in \ell$ from $\ell_1$ to $\ell_2$, 
the angle of 
the edge incident to $p$ on $\Pi (u,p)$
will monotonically increase until it reaches $\pi / 2$, then monotonically decrease.
Thus, after computing the angle of 
the extension segment with $\ell$,
we know to which side of 
its marker
to continue our search:
the side that contains the smaller angle 
(because moving in this direction will increase the smaller angle).
Thus by \cref{lemma: logr time interval search}
the search takes $O(\log r)$ time and $O(r')$ space. 

At the end of our search 
we will have the reflex vertex whose domain contains the edge 
that achieves the angle closest to $\pi/2$.
Then in $O(1)$ time and space we can build the corresponding 
piece of $\operatorname{dist}_u(\cdot)$ and find the value along $\ell$ that
minimizes it.

The space bounds follow from the $n'$ 
projections
that are computed 
and the $O(r')$ space used by the shortest-path data structure queries.
\end{proof}

\begin{figure}[t]
		\includegraphics[width=\linewidth]{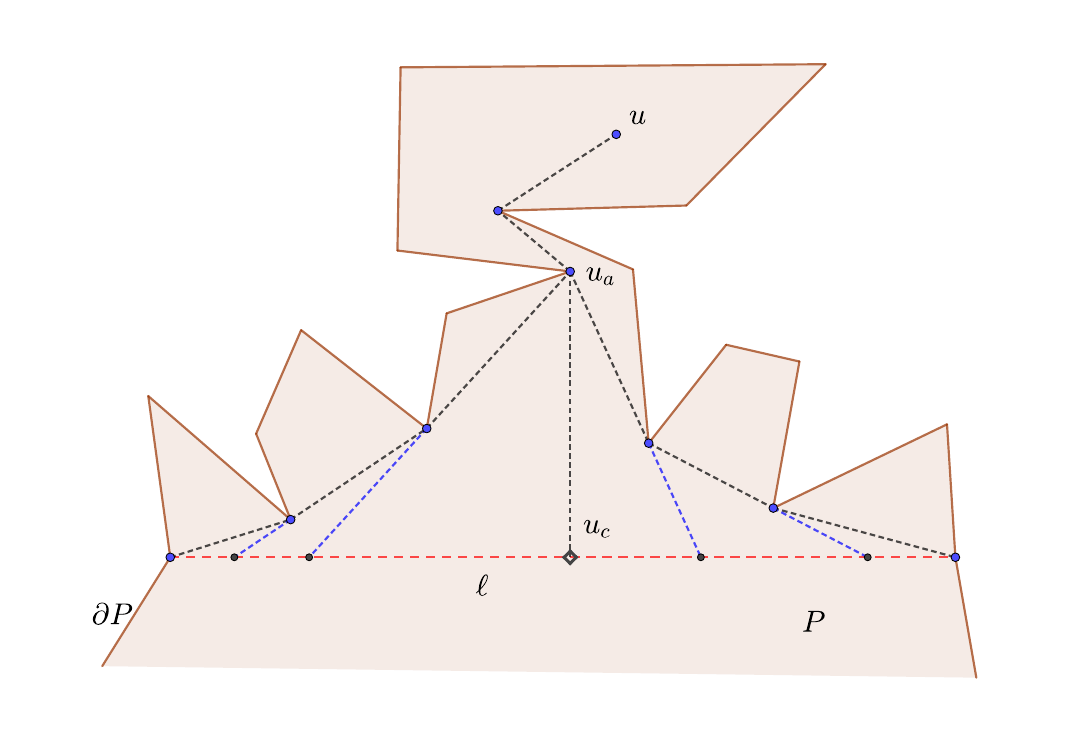}
		\caption{The funnel from $u$ to the endpoints of $\ell$, 
		including the apex $u_a$ and the projection $u_c$ of $u$ 
		onto $\ell$.
		Also seen are the extensions of funnel edges (in blue)
		and their intersection points with $\ell$.
		These intersection points can be used to perform a binary 
		search along $\ell$.} \label{fig: funnel-search}
\end{figure}

\subsection{Intersecting the Chord with Discs}
\label{subsec: decision prob par s}
Let $\partial D(u,\rho)$ be the boundary of
the disc $D(u,\rho)$
for a $u \in S_{\tau}$.
Consider the funnel of $u \in S_{\tau}$ and $\ell$
and its associated domain markers.
Since domain markers are points along $\ell$, 
they provide distances 
against which to compare.
We can use these distances when searching
for points on $\ell$ that are a specified distance away from 
the point 
$u$
(which is the point used to 
define the funnel and hence the markers).
For any point 
$u \in S_{\tau}$,
the distance to $\ell$ increases monotonically as we move
from its 
projection $u_c$
to the endpoints of $\ell$
\cite{DBLP:journals/dcg/PollackSR89}.
Thus if we are given a radius $\rho$, 
we can use 
two funnels to
perform two binary searches among these markers
between $u_c$ and the endpoints of $\ell$
(using \cref{lemma: logr time interval search}, as with \cref{lemma: logr closest point})
to find the at most two domains which contain 
 $\partial D(u,\rho) \cap \ell$.
If we have a particular radius in mind, it takes constant
time and space to test to which side of a domain marker our 
radius lies (since we know $u_c$, and that the distance from $u$ to $\ell$
between $u_c$ and the endpoints of $\ell$ increases monotonically).

As a subroutine for
\ref{alg: merge-me},
for a given radius $\rho$,
for each $u$ in $\mathbb{S}$ we compute the interval $I(u,\rho) = D(u,\rho) \cap \ell$.

\begin{restatable}{lemma}{compIntervals}
\label{lemma: comp intervals}
In $O(|\mathbb{S}|\log r)$ time
and $O(|\mathbb{S}| + r')$ space
we compute $I(u,\rho)$ for each $u$ in the given set $\mathbb{S}$
using the given radius $\rho$.
\end{restatable}

\begin{proof}
Consider the geodesic disc of radius $\rho$, $D(u,\rho)$.
Since the disc is geodesically convex, if the chord intersects the disc
in only one point, it will be at the projection $u_c$.
If it does not intersect the disc, then at $u_c$ the distance
from $u$ to $\ell$ will be larger than $\rho$.
Otherwise, if the chord intersects the disc in two points, 
$u_c$ splits $\ell$ up into two intervals, each with one 
intersection point
(i.e., each one contains a point of $\partial D(u,\rho) \cap \ell$). 
If $u_c$ is an endpoint of $\ell$, assuming $\ell$ has positive length,
one of these intervals may degenerate into a point, making $u_c$ coincide with
one of the intersection points.

These two intervals to either side of $u_c$ have the property that on one side of the intersection
point contained within, the distance from $u$ to $\ell$ is larger than $\rho$, and
on the other side, the distance is less than $\rho$.
Therefore, if $\partial D(u, \rho)$ does intersect $\ell$ in two points
we can proceed as in the proof of \cref{lemma: logr closest point}:
in $O(\log r)$ time and $O(r')$ space we can build the two funnels of $u$ 
between $u_c$ and the endpoints of $\ell$ (truncated at the apices)
and then perform a binary search in each to locate the domain in which 
a point at distance $\rho$ lies.
We find the subinterval delimited by the domain markers of the 
reflex vertices
wherein the distance from $u$ to $\ell$ changes from being 
more (less) than $\rho$ to being less (more) than $\rho$.
Once we find this domain,
we can compute $\partial D(u,\rho) \cap \ell$ in $O(1)$ time and space.
\end{proof}

Once the intervals $ \{I(u,\rho) \mid u \in \mathbb{S}\}$ are computed,
we compute the overlay of these intervals as well as
with $\mathbb{C}$ and compute the depth of the elements of $\mathbb{C}$
(i.e., how many intervals contain the element in question).
See \cref{fig: int depths}.

\begin{figure}
		\includegraphics[width=\linewidth]{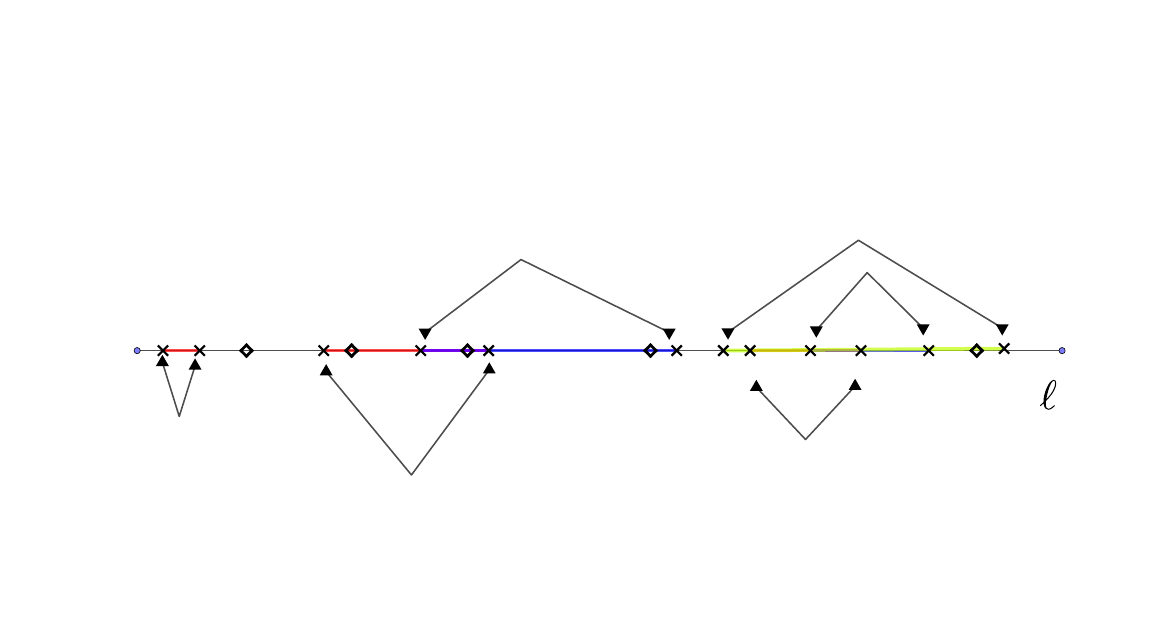}
		\caption{The overlay of intervals 
		$\{I(u,\rho) \mid u \in \mathbb{S}\}$ and $\mathbb{C}$. 
		Intervals (represented with various colours) 
		are delimited by ``$\times$'' symbols with corresponding endpoints 
		matched by conjoined arrows.
		The elements of $\mathbb{C}$ are marked by hollow black diamonds.} \label{fig: int depths}
\end{figure}

\begin{restatable}{lemma}{decisionAlg}
\label{lemma: decision alg}
Given 
$\{I(u,\rho) \mid u \in \mathbb{S}\} \cup \mathbb{C}$, 
in $O( (|\mathbb{S}| + |\mathbb{C}| )\log (|\mathbb{S}| + |\mathbb{C}| ) )$ time
and $O(|\mathbb{S}| + |\mathbb{C}| )$ space:
\begin{itemize}
    \item we compute the overlay $\bigcup_{u \in \mathbb{S}} \{I(u,\rho)\} \cup \mathbb{C}$ of those intervals and $\mathbb{C}$;
    \item we compute the depths of each $w \in \mathbb{C}$.
\end{itemize}
\end{restatable}

\begin{proof}
We are given $O( |\mathbb{S}| )$
labelled intervals along $\ell$, one for each geodesic disc of 
radius $\rho$ centred
on each $u \in \mathbb{S}$.
In other words, the set of these intervals is 
$\{D(u, \rho) \cap \ell : u \in \mathbb{S}\}$.
We sort the points of $\mathbb{C}$ and these interval endpoints
in $O((|\mathbb{S}| + |\mathbb{C}|) \log (|\mathbb{S}| + |\mathbb{C}|))$ time 
and $O(|\mathbb{S}| + |\mathbb{C}|)$ space,
associating each interval endpoint with the interval it opens or closes as well as
the other endpoint for the interval, and the point of $\mathbb{S}$
whose disc created the interval.

When we walk along $\ell$ and enter the interval $D(u, \rho) \cap \ell$ for 
some $u \in \mathbb{S}$,
we say we are in the disc of $u$.
Our last step, done in $O(|\mathbb{S}| + |\mathbb{C}|)$ time and space, 
is to walk along $\ell$ and count the 
number of discs of $\mathbb{S}$ 
we are concurrently in at any given point.
In other words, we count the 
number of overlapping intervals 
(i.e., the depths of the subintervals) for points of $\mathbb{S}$.
As we walk along $\ell$,
when we encounter an element of $\mathbb{C}$
we assign to it the current depth.
\end{proof}

\section{Concluding Remarks}
\label{sec: concluding remarks}

In this paper we have described two methods for computing a 
$2$-SKEG disc
using randomized algorithms. 
With \ref{alg: DI},
we find a $2$-SKEG disc 
using $O(n + m)$ expected space\footnote{The space is considered expected 
because in the base case of the Divide-and-Conquer approach we run the planar approximation
algorithm which expects to use space 
linear in the number of input points. \cite{har2011geometric,har2005fast} }
in $O(n\log^2 n \log r  + m)$ expected time;
and using \ref{alg: RSAlgo},
in deterministic time that is faster
by polylogarithmic factors when $k \in \omega\left(n / \log n\right)$,
we find
a $2$-approximation with high probability 
using worst-case deterministic space $O(n+m)$. 
We leave as an open problem to solve the $2$-SKEG 
problem in $O(n\log r + m)$ time.

At first glance, one may be surprised that the runtime of \ref{alg: DI}
presented in \cref{cor: DC using quickselect} is an expression 
independent of $k$.
One might expect a dependency on $k$ to appear in the runtime of the base case
or the merge step when computing an approximation to a SKEG disc.
In the base case we avoid a dependence on $k$ by using the (expected) 
linear-time $2$-approximation algorithm for planar instances \cite{har2005fast}.
The first few operations in the merge step are projecting the (remaining) candidate points of $S$ 
onto a chord $\ell$, picking a candidate at random, and finding its $k$\textsuperscript{th}-closest 
point from the set $S$.
Here we avoid a $k$ in the runtime by computing (in a brute-force style) all the distances 
from our random candidate to the points of $S$, and then using a linear-time rank-finding
algorithm \cite{DBLP:conf/wea/Alexandrescu17,DBLP:journals/jcss/BlumFPRT73,DBLP:books/daglib/0023376} 
to find the $k$\textsuperscript{th}-smallest distance.
There is an opportunity here to sensitize the runtime of \ref{alg: DI} to $k$ 
(and thus possibly make it more efficient) by finding a clever way to 
find this $k$\textsuperscript{th}-closest point.
The final part of the merge step is to use the discovered distance as a radius
to compute the intersections of geodesic discs centred on points of $S$ with $\ell$,
then perform a plane sweep over the resulting overlapping intervals to count their depth.
However, even if some of the runtimes of these operations are able to be sensitized to $k$ 
it may not help the overall runtime of the merge step
unless a candidate set (our $n$ projections) can be computed in quicker than $O(n\log r)$ time,
and unless the filtering of remaining candidates 
(our intersection of discs with $\ell$ and counting the resulting depths) can be done 
in quicker than $O(n\log n)$ time.

\subsection{Using \texorpdfstring{$k$}{k}-Nearest Neighbour Queries}
\label{subsec: knn disc}

In this section, we discuss modifying \ref{alg: RSAlgo} and \ref{alg: DI}
to use the $k$-nearest neighbour query data structure of de~Berg and Staals
\cite{DEBERG2023101976} and compare the runtimes of these modified versions to 
the original runtimes presented above.

\subsubsection{\ref{alg: RSAlgo}:}
\label{subsubsec: knn rs}

For \ref{alg: RSAlgo},
rather than explicitly computing the distances from the points
in the random sample to the rest of the sites to find the 
$(k-1)$\textsuperscript{st}-closest neighbour in $S$, we can use the 
$k$-nearest neighbour query data structure of de~Berg and Staals
\cite{DEBERG2023101976}.

\begin{restatable}{theorem}{kNNRSAlgoThm}
\label{theorem: knn rs algo}
Using the $k$-nearest neighbour query data structure,
\textbf{\ref{alg: RSAlgo}} computes
a $2$-SKEG disc with 
high probability
in $O(n\log n \log^2 r + n\log ^3 r + m +
(n /k)\log n \log (n+r) \log r)$ expected time
using $O(n\log n \log m + m)$ expected space.
\end{restatable}

When might \cref{theorem: knn rs algo} be expected to be more efficient 
in terms of runtime than
\cref{theorem: rs algo}?
Although \ref{alg: RSAlgo} is used if $k \in \omega(n/ \log n)$,
we will simplify the analysis by comparing the two runtimes when $k= n / \log n$.
In this case,
\cref{theorem: rs algo} runs in $O(n\log^2 n \log r + m)$ time,
and \cref{theorem: knn rs algo} runs in expected
$O(n\log n \log^2 r + n\log ^3 r + m +
\log^2 n \log (n+r) \log r)$ time.
If $\log n > \log r$, we have that $\log (n+r) \in O(\log n)$ and
\[
n\log^2 n \log r > (n\log n \log^2 r + n\log ^3 r) 
\] 
Therefore, 
\[
\log^2 n \log (n+r) \log r \in O(\log^3 n  \log r) \in O(n\log^2 n \log r)
\]
Thus, when $k= n / \log n$ and $\log n > \log r$,
\cref{theorem: knn rs algo} is expected to be more efficient than \cref{theorem: rs algo}.
If $\log r > \log n$, then $n\log ^3 r > n\log^2 n \log r$, thus
\cref{theorem: rs algo} is expected to be more efficient.
If $k \in \Theta(n)$, \cref{theorem: rs algo} runs in  
$O(n\log n \log r + m)$ time, whereas \cref{theorem: knn rs algo} runs in expected
$O(n\log n \log^2 r + n\log ^3 r + m +
\log n \log (n+r) \log r)$ time.
Thus, \cref{theorem: rs algo} is always expected to be more efficient when $k \in \Theta(n)$.

\subsubsection{\ref{alg: DI}:}
\label{subsubsec: knn di algo}
Omitted details from this section appear in \cref{app: knn queries}.
If we wanted to incorporate the $k$-nearest neighbour query data structure
into \ref{alg: DI}, we would use it in the merge step, 
i.e., in \ref{alg: merge-me} to find the $k$\textsuperscript{th}-nearest neighbour
of our randomly chosen projection.
When might we get a more efficient algorithm using the $k$-nearest neighbour query data structure?
Perhaps we can get some intuition by noticing that although finding 
the $k$\textsuperscript{th}-nearest neighbour of our randomly chosen projection
will be faster, the time for the other steps (such as computing intervals along the chord)
remains unchanged.
We use the same notation as in \cref{section: quickselect merge step}.

The first option is to build the data structure over the whole 
polygon using all points of $S$, meaning that every query will give the $k$-nearest neighbours
from $S$, not just the points in the subpolygon under consideration.
Since we already get a $2$-approximation using only the points of the subpolygon being
considered and since $S$ contains this set, we still get a $2$-approximation 
after the recursion if we use the data structure built on all points of $S$ over all of $P$.
Using this approach, the expected space of the algorithm becomes $O(n\log n \log r + m)$
and the expected runtime becomes 
$O( n\log^2 n \log r + r\log n \log (n+r) \log r  +m + rk\log n \log r)$
which is not
as efficient as the expected $O(n\log^2 n \log r + m)$ time 
presented by \cref{cor: DC using quickselect}.

The second option is to build the data structure for each subpolygon of our recursion tree.
The expected space usage becomes $O(n \log n \log^2 r + m)$.
The expected time to perform \ref{alg: DI} becomes
$O(m + n\log^2 n \log r + n\log n \log^3 r + n \log^4 r + r \log n \log (n + r)\log r +  rk \log n \log r )$ using this approach,
which is not
as efficient as the expected time 
presented by \cref{cor: DC using quickselect}.

\subsection*{Acknowledgements}
The authors thank Pat Morin, Jean-Lou de Carufel, Michiel Smid,  
and Sasanka Roy for helpful discussions as well as anonymous reviewers.

\clearpage

\bibliographystyle{plain}
\bibliography{2-skeg-disc}

\clearpage

\appendix
\label{app}

\section{A Note on the Doubling Dimension of the Geodesic Metric}
\label{app: doubling geodesic}

\begin{figure}
	\begin{subfigure}{0.45\textwidth}
		\includegraphics[width=\linewidth]{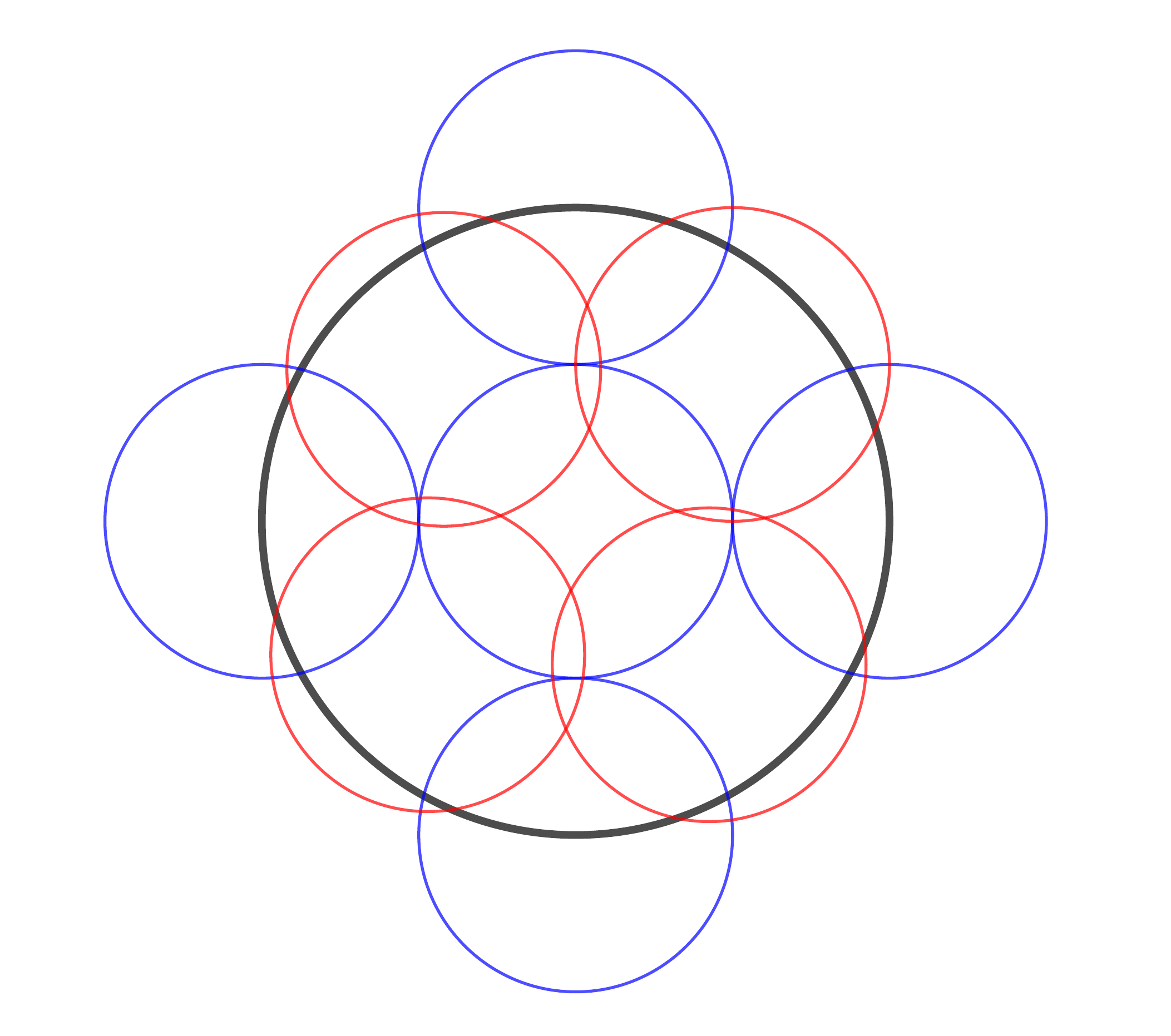}
		\caption{A disc in $\mathbb{R}^2$ covered by nine discs of half its radius.} \label{fig: bounded dd disc}
	\end{subfigure}
	\hspace*{\fill} 
	\begin{subfigure}{0.45\textwidth}
		\includegraphics[width=\linewidth]{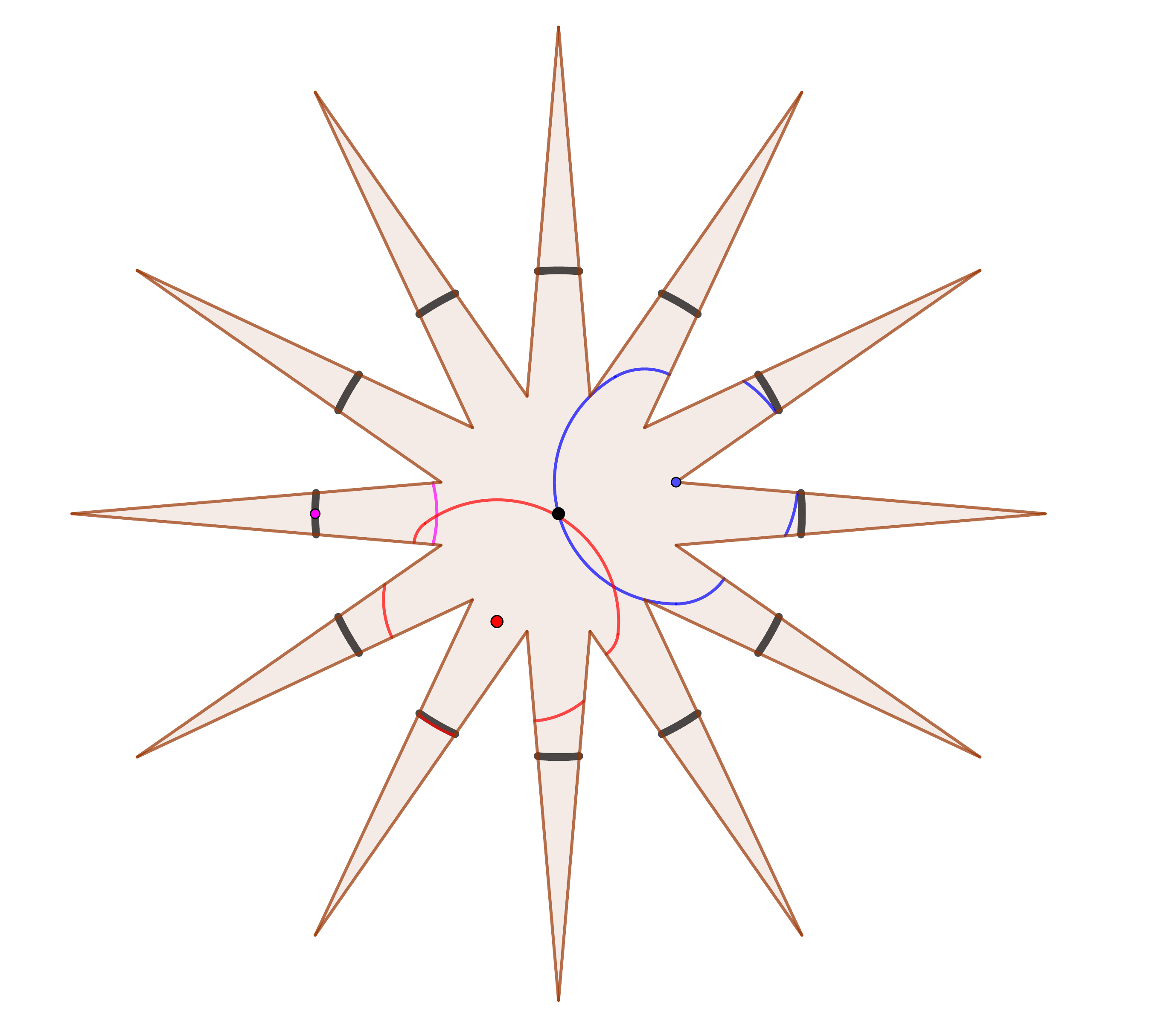}
		\caption{A star-shaped simple polygon with a geodesic disc of radius $1$ and some attempts to cover multiple spikes with geodesic discs of radius $0.5$.} \label{fig: bounded dd geod disc}
	\end{subfigure}
\caption{The larger disc in \cref{fig: bounded dd disc} can be covered by $O(1)$ discs of half its radius, while in \cref{fig: bounded dd geod disc} the geodesic disc of radius $1$ requires $\Omega(r)$ geodesic discs of radius $0.5$ to cover it.}
\label{fig: bounded dd}
\end{figure}

\cref{fig: bounded dd disc} shows a black Euclidean disc of radius $1$ being covered by  
nine red and blue Euclidean discs of radius $0.5$.
On the other hand, \cref{fig: bounded dd geod disc} shows a star-shaped simple polygon
with twelve reflex vertices with symmetric spikes equally-angularly-spaced.
At its centre is a black point from which we draw the black geodesic disc of radius $1$.
The tips of the spikes are at distance $2$ from the centre and the reflex vertices
are at distance $0.5$.
The pink geodesic disc of radius $0.5$ has its pink centre in the left horizontal spike.
This disc centre is in the middle (angularly-speaking)
of the black disc's boundary arc contained in that spike.
The pink disc does not even contain the two reflex vertices of the spike, so using
discs such as this would require at least one disc for each spike.
The red geodesic disc of radius $0.5$, on the other hand, 
has its red centre at such a point as to 
minimize its distance from the black centre while containing the black boundary arc in one of the spikes.
Discs such as this also require at least one disc for each spike since it too covers
the black disc boundary arc contained in just one spike and it cannot
move its centre to cover more of the black disc's boundary without failing to cover
the spike it currently covers.
The blue geodesic disc of radius $0.5$ is centred on a reflex vertex.
It fails to cover the black disc boundary in any spike,
showing that discs of this nature also require at least one disc for each spike.

\clearpage

\section{Figures: Geodesic Bisector Crosses Chord}
\label{app: bisector crosses chord}

\begin{figure}[!h]
		\includegraphics[width=\linewidth]{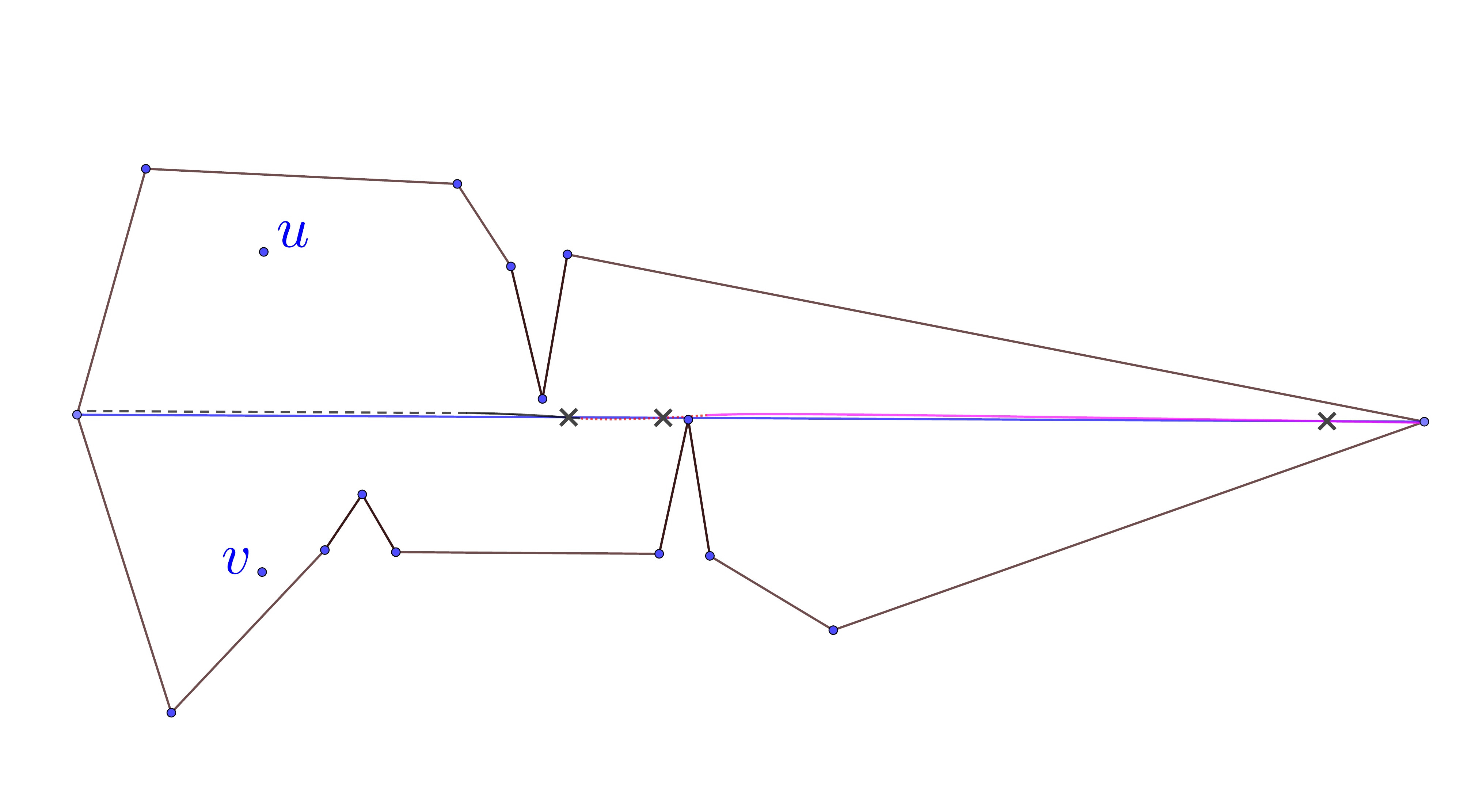}
\caption{The bisector of points $u$ and $v$ on opposite sides of the 
		blue chord of the polygon can intersect the chord 
  multiple times.
		The intersections are labelled with ``$\times$''.
		The different arcs that form the bisector are drawn with 
		different ink
		styles (e.g., dashed vs dotted vs normal ink).}
\label{fig: bisector r crossings}
\end{figure}

\begin{figure}
		\includegraphics[width=\linewidth]{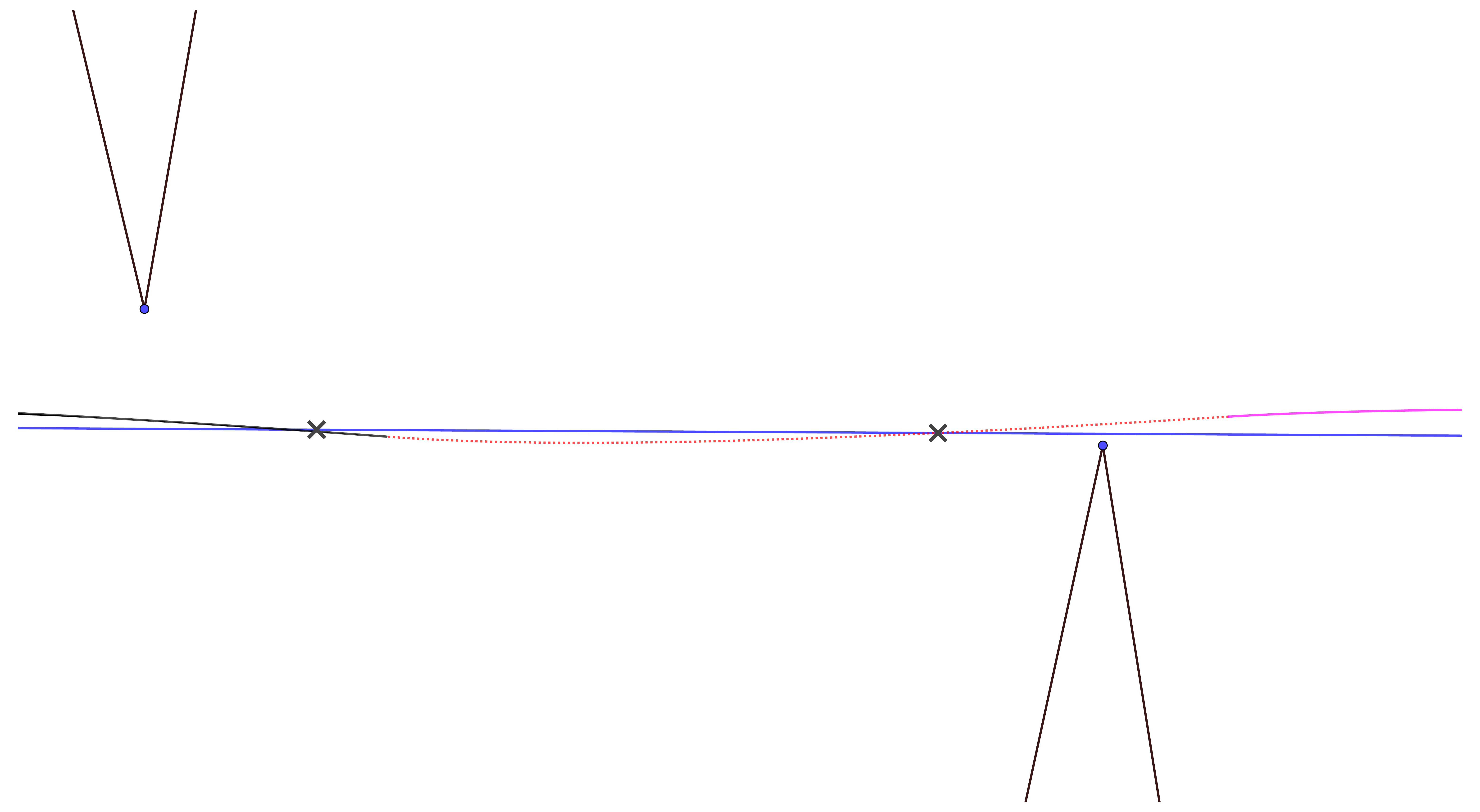}
		\caption{Zoomed-in view of the first two crossings of \cref{fig: bisector r crossings}.}
\label{fig: bisector r crossings zoomed a}
\end{figure}

\begin{figure}
		\includegraphics[width=\linewidth]{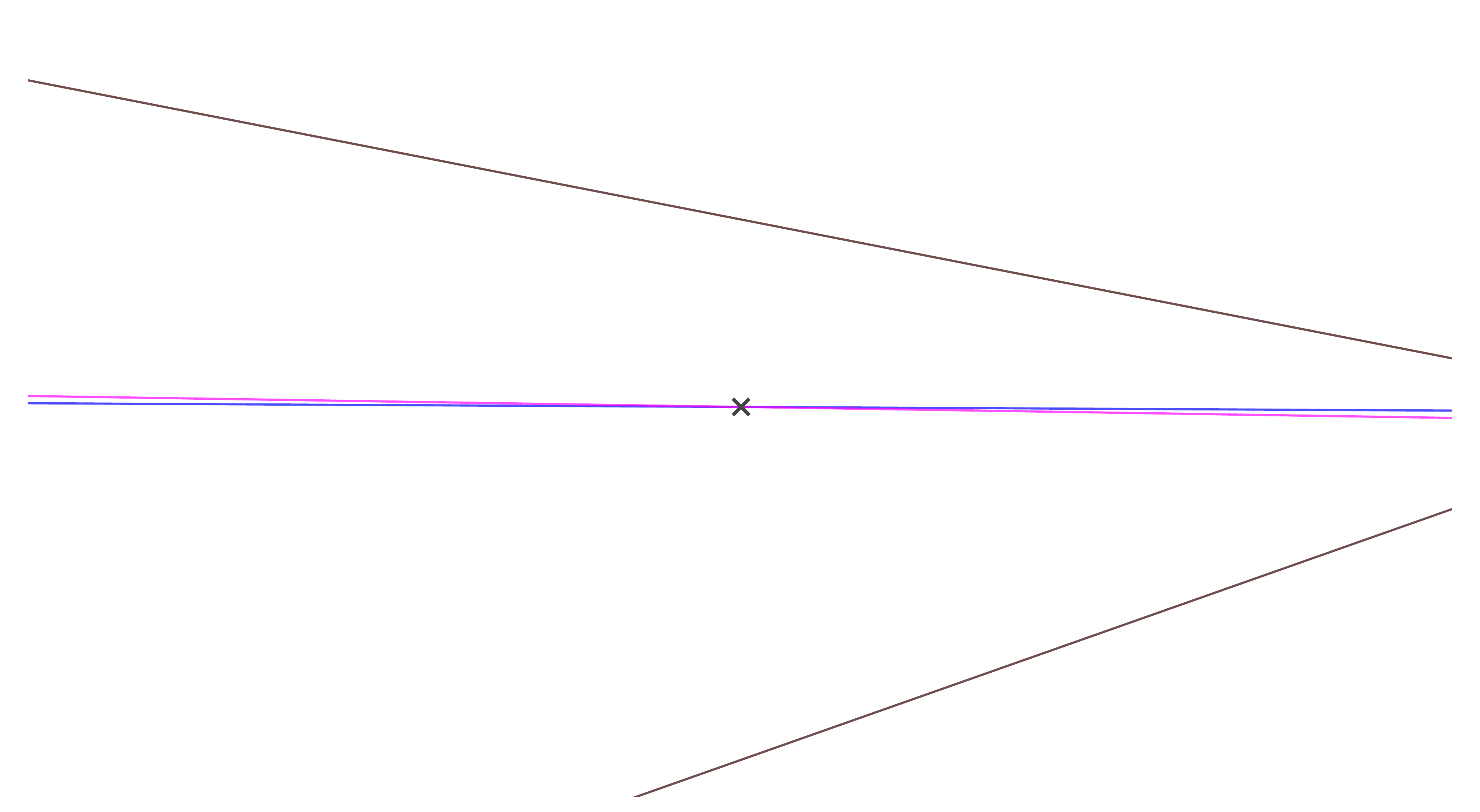}
		\caption{Zoomed-in view of the third crossing of \cref{fig: bisector r crossings}. 
        }
\label{fig: bisector r crossings zoomed b}
\end{figure}

\clearpage

\section{An Exact Algorithm with Voronoi Diagrams}
\label{app: vd}

In this appendix we provide details omitted from \cref{subsec: exact}.

\begin{restatable}{lemma}{okvdCellsEnough}
\label{lemma: okgvd cells enough}
The $k$ points in the SKEG disc define a face of the OKGVD.
\end{restatable}

\begin{proof}
The SKEG disc centred at a point $c^*$ with radius $\rho^*$
has the properties that only the $k$ points of $S$ in the disc 
are within distance $\rho^*$ of $c^*$ 
(by the general position assumption)
and this is a disc with the smallest radius that 
contains $k$ points of $S$.

We show that the $k$ points in the SKEG disc define a face of the OKGVD
by contradiction.
Assume that it is not the case.
A face in the OKGVD has the property that the $k$ points of $S$ defining the face
are the $k$ closest neighbours of $S$ to each point within the face. 
The point $c^*$, which is in the 
polygon, 
must belong to some face $F$ of the OKGVD since:
belonging to an edge or vertex of the OKGVD means there are more than $k$ points in the SKEG disc 
and thus restricting the disc to the $k$ points in just one of the incident faces will allow for a smaller disc;
and because the OKGVD subdivides\footnote{Otherwise we either contradict that we had a simple polygon, 
that there are no two sites equidistant to a vertex of the polygon,
or that there are at least $k$ points of $S$ in the polygon.}
the polygon.
The fact that the points of $S$ defining $F$ are not the same as the set of points
in the SKEG disc contradicts that those points are the $k$ closest points of $S$ to $c^*$.
\end{proof}

\genPosRemark*

\begin{restatable}{lemma}{okvdMinusOneCellsEnough}
\label{lemma: ok minus one gvd cells enough}
If there are three points on the boundary of the SKEG disc,
they define a vertex of the order-$(k-2)$ geodesic Voronoi diagram.
If there are only two points $s, t \in S$
on the boundary of the SKEG disc,
then their geodesic bisector $b(s,t)$
defines an edge of the order-$(k-1)$ geodesic Voronoi diagram
and the position along the edge for the centre of the SKEG disc is the
midpoint along $\Pi(s,t)$.
\end{restatable}

\begin{proof}
The SKEG disc centred at a point $c^*$ with radius $\rho^*$
has the properties that only the $k$ points of $S$ in the disc 
are within distance $\rho^*$ of $c^*$ 
(by the general position assumption)
and this is a disc with the smallest radius 
that contains $k$ points of $S$.

We show that the $k$ points in the SKEG disc are associated with 
a vertex  of the order-$(k-2)$ diagram or an edge of the order-$(k-1)$ diagram.

By our general position assumption, a Voronoi vertex of the 
order-$(k-2)$ Voronoi diagram
(that is not the intersection of a Voronoi edge with the boundary of the polygon) 
has degree three.
This means that there are three cells incident to such a vertex,
and these three cells have $(k-3)$ points in common with three points tied for being
the $(k-2)$\textsuperscript{nd}-closest at the vertex in question
for a total of $k$ points.
Since these vertices have the property that the associated $k$ points are the closest
$k$ points of $S$, if $c^*$ coincided with one of 
these Voronoi vertices and the associated $k$
points were not the $k$ points in the SKEG disc, then we contradict that we had a SKEG disc
(or the proper Voronoi diagram).

If $c^*$ is not a Voronoi vertex of the order-$(k-2)$ diagram, 
then it may lie on a Voronoi edge of the order-$(k-1)$ diagram
and would thus contain the $k$ points defining the cells that share the edge
(or we contradict the set of $k$ points in the SKEG disc).
The two points on the disc boundary are the two points $s,t \in S$ 
whose bisector created the edge.
The position for $c^*$ that minimizes the distance to those two points is the 
midpoint of $\Pi(s,t)$.
If this midpoint does not occur on the Voronoi edge,
then we contradict that we had a SKEG disc:
the disc with only these $k$ points centred at the midpoint has a smaller radius
than a disc containing the same points centred anywhere else along $b(s,t)$;
if the midpoint occurs in some cell, then the 
disc centred at the midpoint containing the specified points
contains at least one point too many (since the points of this other
cell are by definition closer), 
meaning the disc can be shrunk, contradicting that it was a SKEG disc.

If $D^*$ were to have two points on its boundary and 
$c^*$ were 
in a cell (i.e., not on an edge) of the order-$(k-1)$ diagram, 
then the disc centred there would have to contain the points associated
with the cell (otherwise we have a contradiction for what should be contained in a SKEG disc),
but it would also have to contain another point to raise the contained total to $k$.
Since this other point is farther than the $(k-1)$ associated with the cell,
it will be a point on the boundary of the disc.
Since we have a SKEG disc with a minimum radius, the distance to this $k$\textsuperscript{th} point
will be minimized, but that means $c^*$ will be on a bisector between a point
of the cell and this farther point from another cell.
As pointed out above, the position on the bisector that minimizes the radius of the disc
is at the midpoint of the shortest path between the two relevant points.
As pointed out above, this means it is along a Voronoi edge,
contradicting that $c^*$ lies 
in the cell.

A similar argument can be made for $c^*$ being a vertex of
the order-$(k-2)$ diagram given that $D^*$ has three points on its boundary.
\end{proof}

When $k=n$, the SKEG disc problem is the same as the problem of
computing the geodesic centre of a set of points inside a simple polygon.

The order-$1$ and the order-$(n-1)$ geodesic Voronoi diagrams
have complexity $\Theta(n + r)$ \cite{Aronov89,DBLP:journals/dcg/AronovFW93}.
There is an algorithm for the geodesic nearest-point Voronoi diagram
that runs in $O( n\log n + r)$ time and uses 
$O(n\log n + r)$
space 
\cite{DBLP:conf/soda/Oh19}, and one for the
farthest-point diagram that runs in  
$O( n\log n + r)$ time and uses $\Theta(n + r)$ space 
\cite{DBLP:conf/compgeom/Barba19,DBLP:conf/compgeom/000121}.

\cref{lemma: ok minus one gvd cells enough} 
tells us we can use the order-$1$ Voronoi diagram
when $k = 2$.
Considering the preprocessing time, the construction time,
and the time to traverse the diagram performing the 
appropriate $O(\log r)$-time queries, we can solve the SKEG disc problem
in $O(n\log n + n\log r + r\log r + m)$ time.
We say a term in an equation (or the equation) dominates another when 
the former is greater than the latter multiplied by a constant.
We can simplify our expression as follows, assuming the $n\log r$ 
term dominates the others.
This implies the following.

\begin{align*}
&& n\log r & > m \nonumber \\
&& n\log r & > n\log n \nonumber \\
\Rightarrow && \log r & > \log n \nonumber \\
\Rightarrow &&  r & >  n^2 \nonumber \\
\Rightarrow &&  m & >  n^2 \nonumber \\
\Rightarrow &&  m^{1/2} & >  n \nonumber \\
\end{align*}

Using the fact that $m^{1/2} > \log r$, 
this brings us to the contradiction that $m > n\log r$.
Thus, we can simplify the expression to $O(n\log n + r\log r + m)$.

The farthest-point geodesic Voronoi
diagram is the same as the order-$(n-1)$ diagram.
When $k=n-1$,
it is more efficient to compute a SKEG disc using the farthest-point 
geodesic Voronoi diagram and, 
for each face, 
compute the geodesic centre of the $n-1$ closest points.
This gives us the SKEG disc
by \cref{lemma: okgvd cells enough}.

When one uses the OKGVD to compute a SKEG disc,
the approach is to compute the geodesic centres of the $k$ points
defining each face, and then choose the SKEG disc.
Oh and Ahn~\cite[Lemma $3.9$]{Oh2019} presented a method to find
the geodesic centre of $k$ points in 
$O(k\log k \log^2 r)$ time and $O(k+r)$ space.
The other method for computing the appropriate geodesic centres is to 
first compute $CH_g$ for the $k$ points defining each face, and then compute
the geodesic centre of the weakly simple polygon formed by the $CH_g$.

The time for the Oh and Ahn approach (i.e., $O(k\log k \log^2 r)$) always dominates
the time to compute $CH_g$ (i.e., $O(k\log r + k \log k)$).
Thus, the Oh and Ahn approach is only more efficient for 
computing the geodesic centres of the requisite points when its runtime
is dominated by the time to compute the geodesic centre of $CH_g$.
As has already been mentioned, 
the time to compute the geodesic centre of a weakly simple polygon is linear in its
size.
The size of $CH_g$ of the $k$ points of $S$ defining
a face of the OKGVD is $O(r + k)$.
Thus, the Oh and Ahn approach is more efficient when 
$k\log k \log^2 r \in O(r)$.

When $k\log k \log^2 r \in O(r)$,
it is more efficient to use
the approach of Oh and Ahn~\cite[Lemma $3.9$]{Oh2019} to find
the geodesic centre of $k$ points in 
$O(k\log k \log^2 r)$ time and $O(k+r)$ space
for each face of the OKGVD,
yielding  $O((r+n)\cdot k\log k \log^2 r)$ time over all faces
and $O(r+n)$ space.
Substituting $k=n-1$, noting that $k\log k \log^2 r \in O(r)$
means $n \in O(r)$, 
and considering the preprocessing,
the time we use is

\begin{align*}
\begin{split}
	&O((r+n)\cdot n\log n \log^2 r + m)
\end{split}\\
\begin{split}
	=\ &O(nr\log n \log^2 r + m)
\end{split}
\end{align*}

Otherwise, it is more efficient to
compute $CH_g$ of the $k$ points of $S$ associated
with each face of the diagram to obtain weakly simple polygons, 
and then compute the geodesic centres of
those geodesic convex hulls.
Observe that the complexity of $CH_g$ 
of $k$-point subsets of $S$
is $O(r+k)$.
We spend 
$O(r+n)\cdot O(k\log r + k \log k)$
time to compute $CH_g$ of each subset of $k$ points
that defines a face;
then we compute the geodesic centres in
$O(r+n)\cdot O(r+k)$ time
and $O(r+n)$ space.
Substituting
$k = n-1$ and considering the preprocessing cost,
the time used becomes 

\begin{align*}
\begin{split}
	&O((r+n)\cdot (n\log r + n\log n + r) + m)
\end{split}\\
\begin{split}
	=\ &O(n^2\log n	+  n^2\log r +
 nr\log n + nr\log r + r^2 
	+  m)
\end{split}
\end{align*}

Oh and Ahn \cite{Oh2019} showed that the complexity 
of the 
OKGVD
of $n$ points in a simple $r$-gon is
$\Theta(k(n-k) + \min(kr, r(n-k)))$.
The fastest algorithms to compute the OKGVD
are the one by Oh and Ahn~\cite{Oh2019} that runs in time 
$O(k^2n\log n \log^2 r + \min(kr, r(n-k)))$,
and the one by Liu and Lee~\cite{DBLP:conf/soda/LiuL13} that runs in
$O(k^2n\log n + k^2 r \log r)$  
time.
They both use $\Omega(k(n-k) + \min(kr, r(n-k)))$ 
space (which is 
$\Omega(n+r)$ 
for constant $k$ or $k$ close to $n$ 
and 
$\Omega(kn + kr) = \Omega(n^2+nr)$ 
for $k$ a constant fraction of $n$). 

We now show that the approach of Oh and Ahn~\cite{Oh2019}
is more efficient when $n\log n \log r \in O(r)$,
which means the approach of Liu and Lee~\cite{DBLP:conf/soda/LiuL13} 
is more efficient the rest of the time.

Consider the largest-order terms in the running time of the Oh and Ahn approach:

\begin{equation} 
\underbrace{k^2n\log n \log^2 r}_{A} + \underbrace{ \min(kr, r(n-k))}_{B} 
\label{okgvd_ohAhn_time}
\end{equation} 

Similarly, let us consider the largest-order terms in the running time of the 
approach of Liu and Lee:

\begin{equation} 
\underbrace{k^2n\log n}_{C} + \underbrace{k^2 r \log r}_{D} 
\label{okgvd_LL_time}
\end{equation}

When is \cref{okgvd_ohAhn_time} dominated by \cref{okgvd_LL_time}?

Term C is always dominated by term A, so if \cref{okgvd_LL_time} dominates 
\cref{okgvd_ohAhn_time} it will be when term D dominates terms A and B.

When the minimum in term B is $kr$, it is dominated by term D.
When the minimum in term B is $r(n-k)$, it will be $O(nr)$. 
But this happens when $k \in O(n)$ (e.g., $k= n / 2$), in which case
term D becomes $O(n^2 r \log r)$ and dominates term B.

Thus we have that \cref{okgvd_ohAhn_time} is dominated by \cref{okgvd_LL_time}
when:

\begin{align}
&& k^{2}n\log n \log^{2} r & \in O(k^{2}r \log r) \nonumber \\
\Rightarrow && n\log n \log r & \in O(r) \nonumber 
\end{align}

The complexity of the OKGVD
is $\Theta(k(n-k) + \min(kr, r(n-k)))$.
For $1 < k < n-1$,
by \cref{lemma: ok minus one gvd cells enough} it suffices
to compute the order-$(k-2)$ and order-$(k-1)$ geodesic Voronoi diagrams
and consider: the vertices of the order-$(k-2)$ diagram;
and the midpoints of the shortest paths 
between the points defining edges shared by adjacent cells
of the order-$(k-1)$ diagram.
We compute the diagrams in
$O(k^2n\log n \log^2 r + \min(kr, r(n-k)))$ time
for 
$n\log n \log r \in O(r)$ 
and 
$O(k^2 n\log n + k^2 r \log r)$  
time
otherwise.
Using shortest-path queries
we traverse the diagrams in 
$O(k(n-k) + \min(kr, r(n-k)))\cdot O(\log r)$ time
while
comparing the candidates from the edges
and vertices.
Thus, for $n\log n \log r \in O(r)$, 
the time used including preprocessing is

\begin{align}
\begin{split}
	&O(k^2 n\log n \log^2 r + \min(kr, r(n-k))\\
    & \quad + \log r \cdot (k(n-k) + \min(kr, r(n-k)))
	+ m) \nonumber
\end{split}\\
\begin{split}
\label{eqn: first gvd comp a}
	=\ &O(k^2n\log n \log^2 r 
	+ \min(kr, r(n-k))\log r
	+ m) 
\end{split}
\end{align} 

When $k \in O(1)$, we have:

\begin{align*}
	&&& \quad\  O(n\log n \log^2 r 
	+ r\log r
	+ m) \\
\Rightarrow	&&& =  O( r\log r
	+ m) 
\end{align*}

When $k \in \Theta(n)$, we have:

\begin{align*}
\begin{split}
	&O(n^3\log n \log^2 r 
	+ nr\log r
	+ m) 
\end{split}
\end{align*} 

The time for the Liu and Lee approach, including preprocessing, is

\begin{align}
\begin{split}
& O(k^2 n\log n + k^2 r \log r 
 + \log r \cdot (k(n-k) + \min(kr, r(n-k)))
+ m) \nonumber
\end{split}\\
\begin{split}
\label{eqn: second gvd comp a}
=\ &O(k^2 n \log n + k^2 r \log r 
+ k(n-k)\log r +  \min(kr, r(n-k)) \log r
+ m)
\end{split}
\end{align}

When $k \in O(1)$, we get:

\begin{align*}
& O(n \log n + n \log r 
+ r\log r + m)
\end{align*}

When $k \in \Theta(n)$, we have:

\begin{align*}
& O(n^3 \log n + n^2 r \log r + m)
\end{align*}

Compare these runtimes to those we get by 
using the $k$\textsuperscript{th}-order geodesic Voronoi diagram.
When $n\log n \log r \in O(r)$, the more efficient construction time is
$O(k^2n\log n \log^2 r + \min{(kr, r(n-k))})$.
When $k\log k \log^2 r \in O(r)$, the more efficient method for the geodesic centre of $k$ points 
is the method by Oh and Ahn \cite{Oh2019} which runs in $O(k\log k \log^2 r)$ time, which we run on each of the 
$\Theta(k(n-k) + \min{(kr, r(n-k)}))$ faces.
This gives us the following, which we can compare to \cref{eqn: first gvd comp a}:
\begin{align}
\begin{split}
	&O(k^2 n\log n \log^2 r + \min(kr, r(n-k)) \\
  & \quad + k\log k \log^2 r \cdot (k(n-k) + \min(kr, r(n-k)))
	+ m) \nonumber
\end{split}\\
\begin{split}
\label{eqn: first gvd comp b1}
	=\ &O(k^2 n\log n \log^2 r + \min(kr, r(n-k)) + k\log k \log^2 r \cdot \min(kr, r(n-k))
	+ m) 
\end{split}
\end{align} 

When $k \in O(1)$, we have the following.

\begin{align*}
	&&& \quad\  O(n\log n \log^2 r 
	+ r\log^2 r
	+ m) \\
\Rightarrow &&&	=\ O(r \log r 
	+ r\log^2 r
	+ m) \\
\Rightarrow &&&	=\ O( r\log^2 r
	+ m) 
\end{align*}

When $k \in \Theta(n)$, we have the following.

\begin{align*}
\begin{split}
	&O(n^3\log n \log^2 r 
	+ n^2 r\log n \log^2 r
	+ m) 
\end{split}
\end{align*} 

When $k\log k \log^2 r \notin O(r)$, 
the more efficient method of finding the geodesic centre of the $k$ points defining a face is to compute $CH_g$ of those points and then use the linear-time algorithm for the geodesic centre of simple polygons.
Thus, for each of the $\Theta(k(n-k) + \min{(kr, r(n-k)}))$ faces, we compute the $CH_g$ in $O(k \log k + k \log r)$ time, then spend $O(r+k)$ time to find the geodesic centre.
We can compare the following to \cref{eqn: first gvd comp a}.

\begin{align}
\begin{split}
	&O(k^2 n\log n \log^2 r + \min(kr, r(n-k))\\ 
	& \quad + (k\log k + k\log r + r) \cdot (k(n-k) + \min(kr, r(n-k)))
	+ m) \nonumber
\end{split}\\
\begin{split}
\label{eqn: first gvd comp b2}
	=\ &O(k^2n\log n \log^2 r 
     + kr(n-k) +  r \cdot \min(kr, r(n-k))\\
    & \quad + (k\log k + k\log r) \cdot \min(kr, r(n-k))
	+ m) 
\end{split}
\end{align}

The case for $k \in O(1)$ is invalid. 
If $k\log k \log^2 r \notin O(r)$, that implies $r \in O(k\log k \log^2 r)$, 
but because $k \in O(1)$, that implies $r \in O(\log^2 r)$, which is false.

When $k \in \Theta(n)$, we get:

\begin{align*}
\begin{split}
	&O(n^3\log n \log^2 r 
    + n^2r\log n + n^2r\log r
    + nr^2
	+ m) 
\end{split}
\end{align*}

When $n\log n \log r \notin O(r)$, the more efficient OKGVD construction time is
$O(k^2 n\log n + k^2 r\log r)$.
When $k\log k \log^2 r \in O(r)$, we have the following to compare against \cref{eqn: second gvd comp a}:

\begin{align}
\begin{split}
	&O(k^2 n\log n + k^2 r \log r 
	+ k\log k \log^2 r \cdot (k(n-k) + \min(kr, r(n-k)))
	+ m) \nonumber
\end{split}\\
\begin{split}
\label{eqn: second gvd comp b1}
	=\ &O(k^2 n\log n + k^2 r \log r 
	+ k^2(n-k) \log k \log^2 r  \\
   & \quad + k\log k \log^2 r \cdot \min(kr, r(n-k))
	+ m) 
\end{split}
\end{align} 

When $k\in O(1)$, we get:

\begin{align*}
\begin{split}
	&O( n\log n +  n\log^2 r + r\log^2 r
	+ m) 
\end{split}
\end{align*} 

The case for $k \in \Theta(n)$ is invalid.
In this case, we have $r \in O(n\log n \log r)$ and $k \log k \log^2 r \in O(r)$ which, when $k \in \Theta(n)$, implies $n \log n \log^2 r \in O(n\log n \log r)$.

When $k\log k \log^2 r \notin O(r)$, we can compare \cref{eqn: second gvd comp a} to:

\begin{align}
\begin{split}
	&O(k^2 n\log n + k^2 r \log r + (k\log k + k\log r + r) \cdot (k(n-k) + \min(kr, r(n-k)))
	+ m) \nonumber
\end{split}\\
\begin{split}
\label{eqn: second gvd comp b2}
	=\ &O(k^2 n\log n + k^2 r \log r 
    + k^2(n-k)\cdot (\log k + \log r) + kr(n-k)\\
    & \quad+ (k\log k + k\log r) \cdot \min(kr, r(n-k))
    + r \cdot \min(kr, r(n-k))
	+ m) 
\end{split}
\end{align} 

Similar to before, the case for $k \in O(1)$ is invalid; 
we would have the implications that
$r \in O(k\log k \log^2 r)$ and thus $r \in O(\log^2 r)$, which is false.

When $k \in \Theta(n)$, we get:

\begin{align*}
\begin{split}
	&O(n^3 \log n + n^3 \log r + n^2 r\log n + n^2 r\log r + nr^2
	+ m) 
\end{split}
\end{align*} 

Thus, we see that for $1 < k < n-1$
it is more efficient to use the order-$(k-2)$ and order-$(k-1)$ geodesic Voronoi diagrams.

\okvdExactApproach*

\clearpage

\section{DI-Algo Details}
\label{app: di-algo}
In this section we provide details omitted from \cref{sec: Divide et Impera}.

\diAnalysis*

\begin{proof}
As discussed earlier (at the beginning of \cref{sec: Divide et Impera}), 
preprocessing takes $O(n\log r + m)$ time and 
$O(n + m)$ space.

\ref{alg: DI}
is a recursive, Divide-and-Conquer algorithm.
The recurrence tree of the Divide-and-Conquer algorithm
mimics 
$T_B$
and has $O(\log r)$ depth.
The recursive algorithm visits each node of the tree
that represents a subpolygon that contains at least $k$ points of $S$.
For nodes representing subpolygons containing fewer than $k$ points of $S$,
there is no work to be done; such a node and the branch of $T_B$ 
stemming from it are effectively 
pruned from the recursion tree.
By \cref{lemma: merge is 2 approx}, the result of the merge step
is a $2$-approximation for $P_{\tau}$.
Therefore, when we finish at the root, we have a $2$-approximation 
to $D^*$.

The base case of the recursion is triggered when we reach
a leaf of $T_B$.
Here 
$P_{\tau}$
is a triangle.
When the base case is reached,
\ref{alg: DI} 
runs the planar $2$-approximation algorithm 
in expected time and expected space 
linear in the number of points of 
$S_{\tau}$ (i.e., $O(n')$)
\cite{har2011geometric,har2005fast}.

We assume for the moment that the merge step
runs in expected time
$O(n' \log n' \log r + n' \log^2 n')$ 
and uses $O(n' + r')$ space
(we show this in \cref{subsec: qs alg descr}).
The expected running time of the base case is
dominated by the time for the merge step.
As we show below,
this implies that the algorithm runs in 
$O(n \log n \log^2 r + n \log^2 n \log r + m)$ expected time.

Let $\sigma$ be the number of nodes in $T_B$ at the current
level of the tree
(i.e., the number of nodes whose depth in the tree is the same as that of $\tau$).
For $1 \leq j \leq \sigma$, denote the nodes of this level by
$\tau_j$, 
and denote by $n_j$ the number of points of $S$ in the polygon 
associated with $\tau_j$.
For $1 \leq j \leq \sigma$,
for each $\tau_j$
we perform an iteration of the Divide-and-Conquer algorithm
(i.e., the merge step)
in expected time $O(n_j \log n_j \log r + n_j \log^2 n_j)$. 
The runtime depends on $n_j$, which may change from node to node.
However, across a given level of the decomposition tree, the sum of
the different values of $n_j$ is $n$.
Let $X_j$ be the random variable denoting the runtime
of the merge step and all non-recursive operations
for the $j\textsuperscript{th}$ node.
Now let $W$ count the work done in this level of the decomposition tree.
We have the following calculation for the time and expected time
spent on each level of the decomposition tree:

\begin{align*}
W & = O(\Sigma_{j=1}^{\sigma}  X_j )\\
E[W] & = O(E[\Sigma_{j=1}^{\sigma}  X_j ])\\
& = O(\Sigma_{j=1}^{\sigma} E[ X_j ])\\
& = O(\Sigma_{j=1}^{\sigma} (n_{j} \log r \log n_{j} + n_{j}\log^{2} n_{j} ) ) \qquad \text{(by~\cref{lemma: simple rand works})} \\ 
& = O( \log r \Sigma_{j=1}^{\sigma} n_{j}\log n + \Sigma_{j=1}^{\sigma} n_{j}\log^2 n)\\
& = O(\log n \log r (\Sigma_{j=1}^{\sigma} n_{j}) + \log^2 n (\Sigma_{j=1}^{\sigma} n_{j}))\\
& = O(n \log n \log r + n \log^2 n)
\end{align*}

The space bound follows from the space for preprocessing
and the space in the merge step 
which is released after the merge.
Since the expected time spent at each level is 
$O(n \log n \log r + n \log^2 n)$ and there are $O(\log r)$
levels, the expected time bound of the whole algorithm is
$O(n \log n \log^2 r + n \log^2 n \log r)$ 
plus $O(r)$ time to traverse $T_B$.

We arrive at $O(n\log n \log^2 r + n\log^2 n \log r + m)$ expected time,
including the time for preprocessing.
We can simplify this expression.
Consider the largest-order terms in the running time:

\begin{equation} 
\underbrace{n\log n\log^2 r}_{A} + \underbrace{n\log^2 n\log r}_{B} 
+ \underbrace{m}_{C} 
\label{expr1}
\end{equation} 

Assume A dominates B and C, by which we mean A is greater than a constant multiplied by either term.
This implies 
Properties~\eqref{prop2} and~\eqref{prop1}:

\begin{align}
&& n\log n\log^2 r & >   m \label{prop2} \\
&& n\log n\log^2 r & > 3 \cdot (n\log^2 n\log r) \label{prop1} \\
&& \log r & >  3 \cdot (\log n) & \text{by~\eqref{prop1}} \label{prop1b} \\
\Rightarrow && r & > n^3 \nonumber \\
\Rightarrow && m & > n^3 \nonumber\\
\Rightarrow && m^{1/2} & > n^{3/2} \nonumber \\
\Rightarrow && m^{1/2} & > n\log n \label{prop1c} \\
&& m^{1/2} & > \log^2 m \nonumber \\
\Rightarrow && m^{1/2} & > \log^2 r \label{prop3} \\
&& m & > n\log n \log^2 r & 
\text{by~\eqref{prop1c} and~\eqref{prop3}}\label{prop4}
\end{align}
Property~\eqref{prop4} contradicts 
Property~\eqref{prop2}, and our assumption must be false.
Consequently, we can simplify Expression~\eqref{expr1} as follows:
\begin{equation} 
n\log^2 n\log r + m
\label{expr2}
\end{equation} 
\end{proof}

\clearpage

\section{Using \texorpdfstring{$k$}{k}-Nearest Neighbour Queries}
\label{app: knn queries}

\subsection*{\ref{alg: DI}:}
\label{app subsubsec: knn di algo}
If we wanted to incorporate the $k$-nearest neighbour query data structure
into \ref{alg: DI}, we would use it in the merge step, 
i.e., in \ref{alg: merge-me} to find the $k$\textsuperscript{th}-nearest neighbour
of our randomly chosen projection.
When might we get a more efficient algorithm using the $k$-nearest neighbour query data structure?
Perhaps we can get some intuition by noticing that although finding 
the $k$\textsuperscript{th}-nearest neighbour of our randomly chosen projection
will be faster, the time for the other steps (such as computing intervals along the chord)
remains unchanged.
We use the same notation as in \cref{section: quickselect merge step}.

The first option is to build the data structure over the whole 
polygon using all points of $S$, meaning that every query will give the $k$-nearest neighbours
from $S$, not just the points in the subpolygon under consideration.
Since we already get a $2$-approximation using only the points of the subpolygon being
considered and since $S$ contains this set, we still get a $2$-approximation 
after the recursion if we use the data structure built on all points of $S$ over all of $P$.
Using this approach, the expected space of the algorithm becomes $O(n\log n \log r + m)$
and the expected runtime becomes 
$O( n\log^2 n \log r + r\log n \log (n+r) \log r  +m + rk\log n \log r)$.
Indeed, we have the same $O(m + n\log^2 n \log r)$ as before, but now each
node of our decomposition tree will be making this $k$-nearest neighbour query 
$O(\log n')$ times with high probability when not in the base case.
More specifically, in one iteration of the merge step we spend expected 
$O(n'\log n' \log r + n'\log^2 n' + \log n' \log (n+r) \log r + k\log n' \log r)$ 
time.
We can proceed as in the proof of \cref{cor: DC using quickselect}
in \cref{app: di-algo}.

Let $\sigma$ be the number of nodes in $T_B$ at the current
level of the tree
(i.e., the number of nodes whose depth in the tree is the same as that of $\tau$).
For $1 \leq j \leq \sigma$, denote the nodes of this level by
$\tau_j$, 
and denote by $n_j$ the number of points of $S$ in the polygon 
associated with $\tau_j$.
For $1 \leq j \leq \sigma$,
for each $\tau_j$
we perform an iteration of the Divide-and-Conquer algorithm
(i.e., the merge step)
in expected time 
$O(n_j\log n_j \log r + n_j\log^2 n_j + \log n_j \log (n+r) \log r + k\log n_j \log r)$. 
The runtime depends on $n_j$ which may change from node to node.
However, across a given level of the decomposition tree, the sum of
the different values of $n_j$ is $n$.
Let $X_j$ be the random variable denoting the runtime
of the merge step and all non-recursive operations
for the $j\textsuperscript{th}$ node.
Now let $W$ count the work done in this level of the decomposition tree.
Since the first two summands are the same as in the proof of \cref{cor: DC using quickselect},
we omit them and focus on the latter two.
We have the following calculation for the time and expected time
spent on each level of the decomposition tree:

\begin{align*}
W & = O(\Sigma_{j=1}^{\sigma}  X_j )\\
E[W] & = O(E[\Sigma_{j=1}^{\sigma}  X_j ])\\
& = O(\Sigma_{j=1}^{\sigma} E[ X_j ])\\
& = O(\Sigma_{j=1}^{\sigma} (\log n_j \log (n+r) \log r + k\log n_j \log r) ) \\
& = O( \log (n+r)\log r (\Sigma_{j=1}^{\sigma} \log n_j) + k\log r (\Sigma_{j=1}^{\sigma} \log n_j))\\
& = O( \log (n+r)\log r (\Sigma_{j=1}^{\sigma} \log n) + k\log r (\Sigma_{j=1}^{\sigma} \log n))\\
& = O( \log n \log (n+r)\log r (\Sigma_{j=1}^{\sigma} 1) + k\log n \log r (\Sigma_{j=1}^{\sigma} 1))\\
& = O( r \log n \log (n+r)\log r + rk\log n \log r )\\
& \qquad \text{(since there are at most $r$ nodes in a level)}
\end{align*}

However, since there are only $O(r)$ nodes in the entire tree,
this implies that we expect to spend $O( r \log n \log (n+r)\log r + rk\log n \log r )$ time
overall for $k$-nearest neighbour queries, thus we expect to spend
$O(n\log^2 n \log r + m + r\log n \log (n+r) \log r  + rk\log n \log r)$ time, which is not
as efficient as the expected $O(n\log^2 n \log r + m)$ time 
presented by \cref{cor: DC using quickselect}.

The second option is to build the data structure for each subpolygon of our recursion tree.
We use familiar $r_j$ notation to denote the number of subpolygon vertices in 
the $j$\textsuperscript{th} node of a level of $T_B$.

\begin{align*}
E[\text{space for a level of $T_B$}] & = O(E[\Sigma_{j=1}^{\sigma}  \text{space for the $j$\textsuperscript{th} node} ])\\
& = O(\Sigma_{j=1}^{\sigma} E[ \text{space for the $j$\textsuperscript{th} node}  ])\\
& = O(\Sigma_{j=1}^{\sigma} (n_j \log n_j \log r_j ) ) \\
& = O(\Sigma_{j=1}^{\sigma} (n_j \log n \log r ) ) \\
& = O(\log n \log r \Sigma_{j=1}^{\sigma} (n_j ) ) \\
& = O(n \log n \log r ) 
\end{align*}
Given that there are $O(\log r)$ levels of $T_B$ and considering the other space used
in preprocessing, the expected space usage becomes $O(n \log n \log^2 r + m)$.
The expected time to build the $k$-nearest neighbour query data structures for each level
of $T_B$ is as follows.

\begin{align*}
E[\text{construction time for a level of $T_B$}] & = O(E[\Sigma_{j=1}^{\sigma}  \text{time for the $j$\textsuperscript{th} node} ])\\
& = O(\Sigma_{j=1}^{\sigma} E[ \text{time for the $j$\textsuperscript{th} node} ])\\
& = O(\Sigma_{j=1}^{\sigma} n_j (\log n_j \log^2 r + \log^3 r) ) \\
& = O((\log n \log^2 r + \log^3 r) \Sigma_{j=1}^{\sigma} (n_j ) ) \\
& = O( n(\log n \log^2 r + \log^3 r)) 
\end{align*}
Thus we expect construction over the whole tree to take $O( n\log n \log^3 r + n \log^4 r)$ time.
Similar to before, the expected time spent performing queries to the data structure is as follows.

\begin{align*}
E[\text{query time for a level of $T_B$}] & = O(E[\Sigma_{j=1}^{\sigma}  \text{query time for the $j$\textsuperscript{th} node} ])\\
& = O(\Sigma_{j=1}^{\sigma} E[ \text{query time for the $j$\textsuperscript{th} node}  ])\\
& = O(\Sigma_{j=1}^{\sigma} ( \log n_j \log (n_j + r)\log r + k \log n_j \log r ) ) \\
& = O(\Sigma_{j=1}^{\sigma} \log n_j \log (n_j + r)\log r +  \Sigma_{j=1}^{\sigma} k \log n_j \log r ) \\
& = O(\Sigma_{j=1}^{\sigma} \log n \log (n + r)\log r +  \Sigma_{j=1}^{\sigma} k \log n \log r ) \\
& = O(\log n \log (n + r)\log r \Sigma_{j=1}^{\sigma} (1) +  k \log n \log r\Sigma_{j=1}^{\sigma} (1) ) \\
& = O(r \log n \log (n + r)\log r +  rk \log n \log r )
\end{align*}
As before, the $O(r)$ nodes in the tree imply $O(r \log n \log (n + r)\log r +  rk \log n \log r )$
expected time spent performing queries across all nodes of the tree.

Thus the expected time to perform \ref{alg: DI} becomes
$O(m + n\log^2 n \log r + n\log n \log^3 r + n \log^4 r + r \log n \log (n + r)\log r +  rk \log n \log r )$ using this approach,
which is not
as efficient as the expected time 
presented by \cref{cor: DC using quickselect}.

\clearpage

\section{Pseudocode}
\label{app: code}

\ref{alg: RSAlgo} is discussed in \cref{sec: random sampling}.
\begin{algorithm}[h!]
\SetAlgoRefName{RS-Algo}
	\SetKwInOut{Input}{input}\SetKwInOut{Output}{output}
	\LinesNumbered	
	
	\Input{simple polygon $P_{in}$, set $S$ of points in $P_{in}$ with $|S|=n$, $k \leq |S|$}
	\Output{A point $c \in S$ and a value $\rho$ such that $D(c, \rho)$
	is a SKEG disc for points of $S$ that
	is \emph{centred on a point of $S$}}
	\BlankLine
	\tcc{Start preprocessing}
	$P =$ simplifyPolygon($P_{in}$)\;
	Build a shortest-path data structure on $P$\;
	\tcc{End preprocessing}
	minPoint = some point in $S$\;
	minRadius = infinity\;
\For{$i = 0, i < (n / k)\ln (n), ++i$}{ \label{ranFor}
		Pick a point $c \in S$ uniformly at random\; \label{ranPick}
		temp = geodesic distance from $c$ to its $(k-1)^{st}$-closest point in $S\setminus \{c\}$\; \label{bfpointdisc}

		\If{temp $<$ minRadius}{
			minRadius = temp\;
			minPoint = $c$\;		
		}
	}
	\Return{minPoint, minRadius}
	\caption{}
	\label{alg: RSAlgo}
\end{algorithm}

\ref{alg: DI preproc} is discussed in \cref{sec: Divide et Impera}.

\begin{algorithm}[!h]
\SetAlgoRefName{Preproc-Algo}
	\SetKwInOut{Input}{input}\SetKwInOut{Output}{output}
	
	\Input{simple polygon $P_{in}$, set of points $S$ in $P_{in}$ with $|S|=n$, $k \leq |S|$}
	\Output{everything computed below}
	\BlankLine
	$P =$ simplifyPolygon($P_{in}$)\;
	Build a shortest-path data structure on $P$\;
	$T_B=$ a balanced hierarchical polygon decomposition of $P$\;	
	Build an $O(\log r)$ query-time point-location data structure on $P$
	using the triangulation underlying $T_B$\;
	Compute which points of $S$ lie in each subpolygon represented by $T_B$\;
    \caption{}
	\label{alg: DI preproc}
\end{algorithm}

\clearpage

\ref{alg: DI} is presented in \cref{subsec: DI alg descr}.

\begin{algorithm}[!h]
\SetAlgoRefName{DI-Algo}
	\SetKwInOut{Input}{input}\SetKwInOut{Output}{output}
	
	\Input{current node $\tau$ in the balanced decomposition tree $T_B$
	containing diagonal $\ell$, simple polygon $P_{\tau}$ split by $\ell$, set of points $S_{\tau}$ in $P_{\tau}$ with $|S|=n$, $k \leq |S|$}
	\Output{a point $c$ in $P_{\tau}$ and a value $\rho$ such that $D(c, \rho)$
	is a 
    $2$-SKEG disc
	for the	
	points of $S_{\tau}$}
	\BlankLine
	\tcc{Preprocessing done in \ref{alg: DI preproc} before the 
	first recursive call}

    \If{$P_{\tau}$ is a triangle}{
		(minPoint, minRadius) = output of the planar approximation algorithm  
		                        for the points of $S_{\tau}$ in the Euclidean plane \cite{har2011geometric,har2005fast}\;
	} 
	\Else{
		Choose a side of $\ell$ to be \emph{left}\;
		$P_{1}$ = subpolygon of $P_{\tau}$ left of $\ell$\; 
		$P_{2}$ = subpolygon of $P_{\tau}$ right of $\ell$\;
		$S_1$ = points of $S_{\tau}$ in $P_1$\;
		$S_2$ = points of $S_{\tau}$ in $P_2$\;		
		(leftPoint, leftRadius) = DI-Algo(left child of $\tau$ in $T_B$, $P_{1}$, $S_1$, $k$)\;
		(rightPoint, rightRadius) = DI-Algo(right child of $\tau$ in $T_B$, $P_{2}$, $S_2$, $k$)\;		

		\If{leftRadius $<$ rightRadius}{
			minRadius = leftRadius\;
			minPoint = leftPoint\;
		}
		\Else{
			minRadius = rightRadius\;
			minPoint = rightPoint\;			
		}
		\tcc{Begin merge step}
        (tempPoint, tempRadius) = Compute a merge disc centred on $\ell$ for the points of $S_1 \cup S_2$\;
			\If{tempRadius $<$ minRadius}
			{
				minRadius = tempRadius\;
				minPoint = tempPoint\;
			}
		\tcc{End merge step}
	}
	\Return{minPoint, minRadius}

    \caption{}
	\label{alg: DI}
\end{algorithm}

\clearpage

\ref{alg: main} is presented at the end of \cref{sec: Divide et Impera}.

\begin{algorithm}[h!]
\SetAlgoRefName{Main-Algo}
	\SetKwInOut{Input}{input}\SetKwInOut{Output}{output}
	
	\Input{simple polygon $P_{in}$, set of points $S$ in $P_{in}$ whose cardinality is $n$, $k \leq |S|$}
	\Output{a point $c$ in $P_{in}$ and a value $\rho$ such that $D(c, \rho)$
	is a 
    $2$-SKEG
	disc for the	
	points of $S$}
	\BlankLine
	\tcc{Start preprocessing}
	Run \ref{alg: DI preproc} and get polygon $P$\;
	\tcc{End preprocessing}
	\If{$P$ is convex}{
        (minPoint, minRadius) = output of the planar approximation algorithm \cite{har2011geometric,har2005fast}\;
	} 
	\ElseIf{$k \in \omega(n/\log n)$}{
		(minPoint, minRadius) = \ref{alg: RSAlgo}($P$,$S$,$k$)\;
	}
	\Else{
		(minPoint, minRadius) = \ref{alg: DI}(root of $T_B$, $P$, $S$, $k$)\;
	}
	\Return{minPoint, minRadius}

    \caption{}
	\label{alg: main}
\end{algorithm}

\clearpage

\ref{alg: merge-me} is presented in \cref{section: quickselect merge step}.

\begin{algorithm}[htbp]
\SetAlgoRefName{Merge-Algo}
	\SetKwInOut{Input}{input}\SetKwInOut{Output}{output}
	
	\Input{current diagonal $\ell$ from the node $\tau$ in $T_B$ splitting polygon $P_{\tau}$ into $P_1$ and $ P_2$ with $|P_{\tau}| = r'$, 
	set of points $S_{\tau} = S_1 \cup S_2$ in $P_{\tau}$ with $|S_{\tau}|=n'$, $k \leq |S_{\tau}|$}
	\Output{a point $c \in \ell$ and a value $\rho$ such that $D(c, \rho)$ is a 
    $2$-SKEG
	disc for the	
	points of $S_{\tau}$ in $P_{\tau}$ if $D^*$ contains at least one point of $S_1$ and at least one point of $S_2$}
	\BlankLine
	
    Compute $S^{c}_{\tau}$\;
    Initiate candidate set $\mathbb{C} = S^{c}_{\tau}$\;
    Initiate set $\mathbb{S} = S_{\tau}$\;
    
    \While{$|\mathbb{C}| > 0$}
    {
        Pick a point $u_c \in \mathbb{C}$ uniformly at random\;
        Find the point $z \in \mathbb{S}$ that is the $k$\textsuperscript{th}-closest neighbour of $u_c$\;
        Let $\rho = d_g(u_c, z)$\;
        \If{$|\mathbb{C}| == 1$}
        {
            break\;
        }
        \ForEach{$v \in \mathbb{S}$}
        {
            Compute $I(v, \rho)$\;
            \If{$I(v, \rho)$ is empty}
            {
                Remove $v$ from $\mathbb{S}$\;
            }
        }
        \ForEach{$w \in \mathbb{C}$}
        {
            Compute the depth of $w$\;
            \If{the depth of $w$ is less than $k$}
            {
                Remove $w$ from $\mathbb{C}$\;
            }
        } 

        Remove $u_c$ from $\mathbb{C}$\;
    }

	\Return{$u_c$, $\rho$}
    \caption{}
	\label{alg: merge-me}
\end{algorithm}

\end{document}